\documentclass{article}

\PassOptionsToPackage{round}{natbib}

\usepackage[final]{neurips_2025}

\usepackage{amsmath,amsfonts,bm}

\def\eqref#1{equation~\ref{#1}}

\def\1{\bm{1}}

\DeclareMathAlphabet{\mathsfit}{\encodingdefault}{\sfdefault}{m}{sl}
\SetMathAlphabet{\mathsfit}{bold}{\encodingdefault}{\sfdefault}{bx}{n}

\def\sD{{\mathbb{D}}}

\newcommand{\E}{\mathbb{E}}

\usepackage[utf8]{inputenc} 
\usepackage[T1]{fontenc}    
\usepackage{url}            
\usepackage{booktabs}       
\usepackage{amsfonts}       
\usepackage{nicefrac}       
\usepackage{microtype}      

\usepackage[table,xcdraw]{xcolor}
\definecolor{mydarkblue}{rgb}{0,0.08,0.45} 
\usepackage[
    pagebackref,breaklinks,colorlinks,allcolors=mydarkblue
]{hyperref}  
\usepackage{multirow}
\usepackage{subcaption}
\usepackage{amsmath}
\usepackage{amsthm}
\usepackage{graphicx} 
\usepackage{adjustbox}
\DeclareMathOperator{\Var}{Var}

\newcommand{\tp}{\textsc{tp}}

\newcommand{\p}{\textsc{p}}
\newcommand{\tpr}{\textsc{tpr}}
\newcommand{\fpr}{\textsc{fpr}}

\newcommand{\data}{\mathcal{D}}

\newcommand{\shots}{S}
\newcommand{\classes}{C}
\newcommand{\numbershadowmodels}{M}

\newcommand{\LR}{\mathrm{LR}}
\newcommand{\muin}{\mu_{\mathrm{in}}}
\newcommand{\sigin}{\sigma_{\mathrm{in}}}
\newcommand{\muout}{\mu_{\mathrm{out}}}
\newcommand{\sigout}{\sigma_{\mathrm{out}}}
\newcommand{\tx}{t_{\bm{x}}}
\newcommand{\dtarget}{\mathcal{D}_{\mathrm{target}}}

\newcommand{\lp}{\left(}
\newcommand{\rp}{\right)}

\newcommand{\tz}{t_{\bm{z}}}

\newcommand{\hmu}{\hat{\mu}}
\newcommand{\hsig}{\hat{\sigma}}
\newcommand{\hmuin}{\hat{\mu}_{\mathrm{in}}}
\newcommand{\hmuout}{\hat{\mu}_{\mathrm{out}}}
\newcommand{\hsigin}{\hat{\sigma}_{\mathrm{in}}}
\newcommand{\hsigout}{\hat{\sigma}_{\mathrm{out}}}

\newcommand{\lira}{\mathrm{LiRA}}
\newcommand{\rmia}{\mathrm{RMIA}}
\newcommand{\atpr}{\overline{\tpr}}
\newcommand{\afpr}{\overline{\fpr}}

\newcommand{\He}{\mathrm{He}}
\newcommand{\dr}{\mathrm{d}}

\usepackage{mathrsfs}
\newcommand{\domain}{\mathscr{D}}

\newcommand{\offlira}{\mathrm{offLiRA}}
\usepackage[capitalize,noabbrev]{cleveref}
\usepackage{enumitem}
  \newlist{inlinelist}{enumerate*}{1}
  \setlist*[inlinelist,1]{%
          label=(\roman*),
      }

\usepackage{wrapfig}
\usepackage{thm-restate}
\theoremstyle{plain}
\newtheorem{theorem}{Theorem}

\newtheorem{lemma}[theorem]{Lemma}
\newtheorem{corollary}[theorem]{Corollary}

\theoremstyle{definition}

\theoremstyle{remark}
\newtheorem{remark}[theorem]{Remark}

\usepackage{tocloft}
\title{Impact of Dataset Properties on Membership Inference Vulnerability of Deep Transfer Learning}

\author{
Marlon Tobaben$^{1}$\thanks{These authors contributed equally.} \quad Hibiki Ito$^{2*}$\thanks{Work performed in part while at the University of Helsinki.} \quad Joonas Jälkö$^{1*}$ \quad Yuan He$^{1}$ \quad Antti Honkela$^{1}$\\
$^1$Department of Computer Science, University of Helsinki, Finland \\$^2$School of Informatics, Kyoto University, Japan\\
\texttt{\{marlon.tobaben,joonas.jalko,yuan.he,antti.honkela\}@helsinki.fi}\\
\texttt{ito.hibiki.77n@st.kyoto-u.ac.jp} \\
}

\begin{document}
\addtocontents{toc}{\protect\setcounter{tocdepth}{0}}

\maketitle

\begin{abstract}
Membership inference attacks (MIAs) are used to test practical privacy of machine learning models.
MIAs complement formal guarantees from differential privacy (DP) under a more realistic adversary model.
We analyze MIA vulnerability of fine-tuned neural networks both empirically and theoretically, the latter using a simplified model of fine-tuning.
We show that the vulnerability of non-DP models when measured as the attacker advantage at a fixed false positive rate reduces according to a simple power law as the number of examples per class increases.
A similar power-law applies even for the most vulnerable points, but the dataset size needed for adequate protection of the most vulnerable points is very large.
\end{abstract}

\section{Introduction}\label{sec:introduction}
\begin{wrapfigure}{r}{0.5\textwidth}
\vspace{-4mm}
\begin{minipage}{\linewidth}
    \centering
    \includegraphics[width=1.0\linewidth]{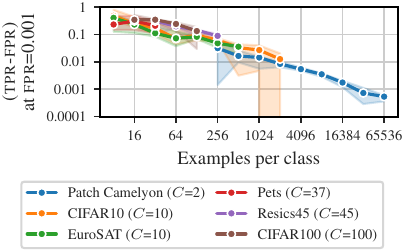}
    \caption{We observe a power-law relation between MIA vulnerability and examples per class (denoted as $S$ or shots) when attacking a fine-tuned ViT-B Head using LiRA. Each colored line denotes a different fine-tuning dataset where $C$ specifies the number of classes. The solid line is median and the error bars the min/max bounds for the Clopper-Pearson CIs over six seeds.}\label{fig:fig1}
\vspace{-6mm}
\end{minipage}
\end{wrapfigure}
Membership inference attacks~\citep[MIAs;][]{Shokri2017MIA,Carlini2022LiRA} and differential privacy~\citep[DP;][]{dwork2006calibrating} provide complementary means of deriving lower and upper bounds for the privacy loss of a machine learning algorithm.
Yet, the two operate under slightly different threat models.
DP implicitly assumes a very powerful adversary with access to all training data except the target point and provides guarantees against every target point.

MIAs assume an often more realistic adversary model with access to just the data distribution and the unknown training data becoming latent variables that introduce stochasticity into the attack. However, the practical evaluation is statistical and cannot provide universal guarantees.

In this paper, we seek to explore MIA vulnerability to extrapolate this gap.
Inspired by an empirical finding that average MIA vulnerability of neural network fine-tuning strongly reduces as the number of samples in the target class increases (see \Cref{fig:fig1}), we develop theory of optimal MIA against a simple model of neural network fine-tuning and reproduce the decrease in vulnerability.
Furthermore, the theoretical model predicts that the vulnerability of all individual samples should reduce as the number of samples increases, which we are able to verify empirically.

To achieve our goal, we theoretically analyze and systematically apply two state-of-the-art black-box MIAs, LiRA \citep{Carlini2022LiRA} and RMIA \citep{Zarifzadeh2024RMIA}, to help understand practical privacy risks when fine-tuning deep-learning-based classifiers without DP protections. For the theoretical model that we analyze, LiRA is the optimal attack by the Neyman--Pearson lemma \citep{NeymanPearson1933}. Under the black-box threat model, in which the adversary does not have access to the model parameters, LiRA and RMIA have been shown to empirically outperform other attacks, especially when the number of shadow models is sufficiently large.

We focus on transfer learning using fine-tuning because this is increasingly used for all practical applications of deep learning and especially important when labeled examples are limited, which is often the case in privacy-sensitive applications.
Our case study focuses on understanding and quantifying factors that influence the vulnerability of non-DP deep transfer learning models to MIA.
In particular, we theoretically study the relationship between the number of examples per class, which we denote as shots ($\shots$), and MIA vulnerability (true positive rate $\tpr$ at fixed false positive rate $\fpr$) for a simplified model of fine-tuning, and derive a power-law relationship (\cref{fig:fig1}) in the form
\begin{align}
    \log(\tpr-\fpr) = -\beta_{\shots}\log(\shots) - \beta_0.
    \label{eq:general_form_power_law}
\end{align}

We complement the theoretical analysis with extensive experiments over many datasets with varying sizes, in the transfer learning setting for image classification tasks, and observe the same power-law.
Based on extrapolation from our results, the number of examples per each class that are needed for adequate protection of the most vulnerable samples appears very high.

\textbf{Related work}
There has been evidence that classification models with more classes are more vulnerable to MIA \citep{Shokri2017MIA}, models trained on fewer samples can be more vulnerable \citep{Chen2020ganleaks,nemeth2025FLMIA}, and classes with less examples tend to be more vulnerable \citep{Chang2021-Subgroup,Kulynych2022_Disparity,Tonni2020-DataMIA}.
A larger generalization error, which is related to dataset size, has also been shown to be sufficient for MIA success \citep{Song2021-GError}, though not necessary \citep{Yeom2018-Overfit}.
Similarly, minority subgroups tend to be more affected by DP \citep{Suriyakumar2021-DPsubgroups,Bagdasaryan2019-DPimpact}.
\citet{Feldman20_NNMemorize} showed that neural networks trained from scratch are required to memorize a significant fraction of their training data to obtain high utility, while the memorization is greatly reduced for fine-tuning.
Additionally, \citet{tobaben2023Efficacy} reported how the MIA vulnerability of few-shot
image classification is affected by the number of shots.
\citet{yu2023individual} studied the relationship between the MIA vulnerability and individual privacy parameters for different classes.
Recently, worst-case MIA vulnerability has gained more attention \citep{Guepin2024-Specific,Meeus2024-Achilles,Azize2025-TargetMIA}.
Nonetheless, the prior works do not consider the rate of change in the vulnerability evaluated at a low $\fpr$, as dataset properties change.
Our work significantly expands on these works by a) explicitly identifying a quantitative relationship between dataset properties and MIA vulnerability, i.e., the power-law in \cref{eq:general_form_power_law}, and b) focusing on the worst-case vulnerability, both evaluated at a low FPR. This in turn allows us to extrapolate MIA vulnerability to DP guarantee.

\textbf{List of contributions}
We analyze the MIA vulnerability of deep transfer learning using two state-of-the-art score-based MIAs, LiRA~\citep{Carlini2022LiRA} and RMIA~\citep{Zarifzadeh2024RMIA}, which are a strong realistic threat model. 
We first analytically derive the power-law relationship in \cref{eq:general_form_power_law} for both MIAs by introducing a simplified model of the optimal membership inference (\cref{sec:explaining_trends}). We support our theoretical findings by an extensive empirical study on the MIA vulnerability of
deep learning models by focusing on a transfer-learning setting for image classification task, where a large pre-trained neural network is fine-tuned on a sensitive dataset.

\begin{enumerate}[leftmargin=*]
            \item \textit{Power-law in a simplified model of the optimal MIA:} We formulate a simplified model of MIA to quantitatively relate dataset properties and MIA vulnerability. In this model LiRA is the optimal attack. For this model, we prove a power-law relationship for both average and worst-case between the LiRA as well as RMIA vulnerability and the number of examples per class (See \cref{sec:simplified-model}).

            \item \textit{MIA experiments on the average case vulnerability:} We conduct a comprehensive study of MIA vulnerability ($\tpr$ at a fixed low $\fpr$) in the transfer learning setting for image classification tasks with target models trained using many different datasets with varying sizes and confirm the theoretical power-law between the number of examples per class and the vulnerability to MIA (see \cref{fig:fig1,sec:experimental_results}). We fit a regression model which follows the functional form of the theoretically derived power-law. We show both a very good fit on the training data as well as a good prediction quality on unseen data from a different feature extractor and when fine-tuning other parameterizations (see \cref{sec:model-to-predict-vulnerability}).
            \item \textit{MIA experiments on the worst-case vulnerability:} We extend the experiments to worst-case individual sample vulnerabilities and observe a similar decrease in vulnerability for quantiles of vulnerable data points and a slower decrease for the maximum individual vulnerability (\cref{sec:individual_mia}).
            By extrapolation we find that an adequate protection of the most vulnerable samples would require an extremely large dataset (\cref{sec:bound_comparision}).
\end{enumerate}

\section{Background}\label{sec:background}
\textbf{Notation} for the properties of the training dataset $\mathcal{\data}$:
\begin{inlinelist}
    \item $\classes$ for the number of classes
    \item $\shots$ for shots (examples per class)
    \item $\mathcal{|\data|}$ for training dataset size ($\mathcal{|\data|}=CS$).
\end{inlinelist}
We denote the number of MIA shadow models with $\numbershadowmodels$.

\textbf{Membership inference attacks (MIAs)} aim to infer whether a particular sample was part of the training set of the targeted model~\citep{Shokri2017MIA}. Thus, they can be used to determine lower bounds on the privacy leakage of models to complement the theoretical upper bounds obtained through DP.

\textbf{Likelihood Ratio attack}~\citep[\textbf{LiRA};][]{Carlini2022LiRA}
While many different MIAs have been proposed~\citep{hu2022membership}, in this work we consider the Likelihood Ratio Attack (LiRA). LiRA is a strong attack that assumes an attacker that has black-box access to the attacked model, knows the training data distribution, the training set size, the model architecture, hyperparameters and training algorithm.
Based on this information, the attacker can train so-called shadow models~\citep{Shokri2017MIA} which imitate the model under attack but for which the attacker knows the training dataset.

LiRA exploits the observation that the loss function value used to train a model is often lower for the examples that were part of the training set compared to those that were not.
For a target tuple $(\bm{x},y)$, where $y$ is a label, LiRA trains the shadow models:
\begin{inlinelist}
\item with $(\bm{x}, y)$ as a part of the training set ($(\bm{x}, y) \in \data$) and
\item without $\bm{x}$ in the training set ($(\bm{x}, y) \notin \data$)
\end{inlinelist}.
After training the shadow models, $(\bm{x}, y)$ is passed through the shadow models, and
based on the losses (or predictions) two Gaussian distributions are formed: 
one for the losses of $(\bm{x}, y) \in \data$ shadow models, and one for the 
$(\bm{x}, y) \notin \data$.
Finally, the attacker computes the loss for the point $\bm{x}$ using the model under
attack and determines using a likelihood ratio test on the distributions built
from the shadow models whether it is more likely that $(\bm{x}, y) \in \data$ or $(\bm{x}, y) \notin
\data$. 
When the true distributions of the shadow models are Gaussians, LiRA is the optimal attack
provided by the Neyman--Pearson lemma \citep{NeymanPearson1933}. 
We use an optimization by \citet{Carlini2022LiRA} for performing LiRA for
multiple models and points without training a computationally infeasible number
of shadow models. It relies on sampling the shadow datasets in a way that each sample is in expectation half of the time included in the training dataset of a shadow model and half of the time not. At attack time each model will be attacked once using all other models as shadow models.

\textbf{Robust Membership Inference Attack}~\citep[\textbf{RMIA};][]{Zarifzadeh2024RMIA} RMIA is a new MIA algorithm, which aims to improve performance when the number of shadow models is limited. We show both theoretically and empirically that the power-law also holds for RMIA.

\textbf{Measuring MIA vulnerability} 
Using the chosen MIA score of our attack,
we can build a binary classifier to predict whether a sample belongs to the 
training data or not. The accuracy profile of such classifier 
can be used to measure the success of the MIA. More specifically, throughout
the rest of the paper, we will use the true positive rate ($\tpr$) at a 
specific false positive rate ($\fpr$) as a measure for the vulnerability.
Identifying even a small number of examples with high confidence is considered 
harmful~\citep{Carlini2022LiRA} and thus we focus on the regions of small $\fpr$. 

\textbf{Measuring the uncertainty for $\tpr$}
The $\tpr$ values from the LiRA-based classifier can be seen as maximum
likelihood-estimators for the probability of producing true positives among the
positive samples. Since we have a finite number of samples for our estimation,
it is important to estimate the uncertainty in these estimators. Therefore, when
we report the $\tpr$ values for a single repeat of the learning algorithm, we
estimate the stochasticity of the $\tpr$ estimate by using Clopper-Pearson
intervals \citep{clopperUseConfidenceFiducial1934}. Given $\tp$ true positives
among $\p$ positives, the $1-\alpha$ confidence Clopper-Pearson interval for the
$\tpr$ is given as 
\begin{equation}
\begin{split}
     & B(\alpha/2; \tp, \p-\tp+1) < \tpr < B(1-\alpha/2; \tp+1, \p-\tp),
\end{split}
\end{equation}
where $B(q; a, b)$ is the $q$th-quantile of $\text{Beta}(a,b )$ distribution.

\section{Theoretical analysis}\label{sec:explaining_trends}
In this section, we seek to theoretically understand the impact of the dataset properties on the MIA vulnerability.
It is known that different data points exhibit different levels of MIA vulnerability depending on the underlying distribution \citep[e.g.\ ][]{Aerni2024-br, Leemann2024-mi}.
Therefore, we start with analysing \emph{per-example} MIA vulnerabilities.
In order to quantitatively relate dataset properties to these vulnerabilities, a simplified model is formulated.
Within this model, we prove a power-law between the per-example vulnerability and the number $S$ of examples per class.
Finally, the per-example power-law is analytically extended to \emph{average-case} MIA vulnerability, for which we provide empirical evidence in \cref{sec:predicting_dataset}.
We primarily focus on the analysis of (online) LiRA, since it is the optimal attack in our simplified model. We show that similar theoretical results also hold for RMIA in \cref{sec:rmia-theory} and offline LiRA \cref{sec:offline-lira}.

\subsection{Preliminaries}\label{sec:notations}
First, let us restate the MIA score from (online) LiRA as defined by \cite{Carlini2022LiRA}. Denoting the logit of a target model $\mathcal{M}$ applied on a target data point $(\bm{x},y)$ as $\ell(\mathcal{M}(\bm{x}), y)$, LiRA computes the MIA score as the likelihood ratio
\begin{align}
    \LR(\bm{x}) = \frac{p(\ell(\mathcal{M}(\bm{x}), y) \mid \mathbb{Q}_{\text{in}}(\bm{x}, y))}{p(\ell(\mathcal{M}(\bm{x}), y) \mid \mathbb{Q}_{\text{out}}(\bm{x}, y))},
\end{align}
where the $\mathbb{Q}_{\text{in}/\text{out}}$ denote the hypotheses that $(\bm{x}, y)$ was or was not in the training set of $\mathcal{M}$. \cite{Carlini2022LiRA} approximate the IN/OUT hypotheses as normal distributions. Denoting $\tx = \ell(\mathcal{M}(\bm{x}), y)$, the score becomes
\begin{align}
    \LR(\bm{x})= \frac{\mathcal{N}(\tx; \hmuin(\bm{x}), \hsigin(\bm{x})^2 )}{\mathcal{N}(\tx; \hmuout(\bm{x}), \hsigout(\bm{x})^2)},
    \label{eq:lira_score}
\end{align}
where the $\hat{\mu}_{\text{in/out}}(\bm{x})$ and $\hat{\sigma}_{\text{in/out}}(\bm{x})$ are the means and standard deviations for the IN/OUT shadow model losses for $(\bm{x}, y)$. Larger values of $\LR(\bm{x})$ suggest that $(\bm{x}, y)$ is more likely in the training set and vice versa. Now, to build a classifier from this score, the LiRA tests if $\LR(\bm{x}) > \tau$ for some threshold $\tau$. Note that the attacker only has a finite set of shadow models to estimate the IN/OUT parameters. Therefore, the MIA scores become random variables over the true population level IN/OUT distributions.

\subsection{Computing the $\tpr$ for LiRA}
Using the LiRA formulation of \Cref{eq:lira_score}, the $\tpr$ for the target point $(\bm{x}, y)$ for LiRA is defined as
\begin{equation}
    \tpr_\lira(\bm{x}) = \Pr_{\dtarget\sim\sD^{|\data|}, \phi^M} \lp \LR(\bm{x}) \geq \tau \mid (\bm{x}, y) \in \dtarget\rp,
\end{equation}
where $\tau$ is a threshold that defines a rejection region of the likelihood ratio test, and $\phi^M$ denotes the randomness in shadow set sampling and shadow model training (see \cref{sec:formulate-lira-rmia} for derivation).

We define the average-case $\tpr$ for LiRA by taking the expectation over the data distribution:
\begin{equation}
    \atpr_\lira = \E_{(\bm{x}, y) \sim \sD}[\tpr_\lira(\bm{x})]
\end{equation}

\subsection{Per-example MIA vulnerability}\label{sec:per-example-MIA}
Although LiRA models $\tx$ by a normal distribution, we consider a more general case where the true distribution of $\tx$ is of the location-scale family.
That is,
\begin{equation}
    \tx = \begin{cases}
        \muin(\bm{x}) + \sigin(\bm{x}) Z &\text{ if } (\bm{x}, y) \in \dtarget \\
        \muout(\bm{x}) + \sigout(\bm{x}) Z &\text{ if } (\bm{x}, y) \notin \dtarget,
    \end{cases}
\end{equation}
where $Z$ has the standard location and unit scale, and $\muin(\bm{x}), \muout(\bm{x})$ and $\sigin(\bm{x}),\sigout(\bm{x})$ are the locations and scales of IN/OUT distributions of $\tx$.
We assume that the target and shadow datasets have a sufficient number of examples.
This allows us to also assume that $\hsig(\bm{x})=\hsigin(\bm{x})=\hsigout(\bm{x})$ and $\sigma(\bm{x})=\sigin(\bm{x})=\sigout(\bm{x})$, where $\hsig(\bm{x})$ is the standard deviation of $\tx$ estimated from shadow models and $\sigma(\bm{x})$ is the true scale parameter of $\tx$. (See \cref{sec:shared-scale} for the validity of these assumptions).
The following result reduces the LiRA vulnerability to the location and scale parameters of $\tx$.
\begin{restatable}[Per-example LiRA vulnerability]{lemma}{lemmaPerExampleLiRA}\label{lemma:per-example-LiRA}
    Suppose that the true distribution of $\tx$ is of location-scale family with locations $\muin(\bm{x}), \muout(\bm{x})$ and scale $\sigma(\bm{x})$, and that LiRA models $\tx$ by $\mathcal{N}(\hmuin(\bm{x}), \hsig(\bm{x})^2)$ and $\mathcal{N}(\hmuout(\bm{x}),\hsig(\bm{x})^2)$.
    Assume that an attacker has access to the underlying distribution $\sD$.
    Then for a large enough number of examples per class and infinitely many shadow models, the LiRA vulnerability of a fixed target example is
    \begin{equation}
        \tpr_\lira(\bm{x}) = \begin{cases}
        1 - F_Z\lp F_Z^{-1}(1 - \fpr_\lira(\bm{x})) - \frac{\muin(\bm{x}) - \muout(\bm{x})}{\sigma(\bm{x})}\rp &\text{ if }\hmuin(\bm{x}) > \hmuout(\bm{x}) \\
        F_Z\lp F_Z^{-1}(\fpr_\lira(\bm{x})) - \frac{\muin(\bm{x}) - \muout(\bm{x})}{\sigma(\bm{x})}\rp &\text{ if }\hmuin(\bm{x}) < \hmuout(\bm{x}),
        \end{cases} \label{eq:lira-px-tpr}
    \end{equation}
    where $F_Z$ is the cdf of $Z$ with the standard location and unit scale, assuming that the inverse of $F_Z$ exists.
\end{restatable}
\begin{proof}
    See \cref{sec:proof-per-example-LiRA}.
\end{proof}
Here we assume that an attacker trains shadow models with the true underlying distribution.
However, in real-world settings the precise underlying distribution may not be available for an attacker.
We relax this assumption in \cref{sec:relax-assumption} so that the attacker only needs an approximated underlying distribution for the optimal LiRA as in \cref{lemma:per-example-LiRA}.

\subsection{A simplified model of the optimal membership inference}\label{sec:simplified-model}
Now we construct a simplified model of membership inference that streamlines the data generation and shadow model training.

We sample vectors on a high-dimensional unit sphere and classify them based on inner product with estimated class mean. This model is easier to analyze theoretically than real-world deep learning examples.
We generate the data and form the classifiers (which are our target models) as follows:
\begin{enumerate}[leftmargin=*]
    \item For each class, we first sample a true class mean $\bm{m}_{c}$ on a high dimensional unit sphere that is orthogonal to all other true class means ($\forall i,j \in \{1,\dots,C\}, i \neq j: \bm{m}_i  \perp  \bm{m}_{j}$).
    \item We sample $2 \shots$ vectors $\bm{x}_c$ for each class. We assume that they are Gaussian distributed around the the true class mean $\bm{x}_c \sim \mathcal{N}(\bm{m}_{c}, \Sigma)$ where the $\Sigma$ is the in-class covariance.
    \item For each ``target model'' we randomly choose a subset of size $\classes \shots$ from all generated vectors and compute per-class means $\hat{\bm{m}}_c$.
    \item The computed mean is used to classify sample $\bm{x}$ by computing
    the inner product $\langle \bm{x}, \hat{\bm{m}}_c \rangle$ as a metric of similarity.
\end{enumerate}

The attacker has to infer which vectors have been used for training the classifier. Instead of utilising the logits (like in many image classification tasks), the attacker can use the inner products of a point with the cluster means. Since the inner product score follows a normal distribution, LiRA with infinitely many shadow models is the optimal attack by the Neyman--Pearson lemma \citep{NeymanPearson1933}, which states that the likelihood ratio test is the most powerful test for a given FPR.

This simplified model resembles a linear (Head) classifier often used in transfer learning when adapting to a new dataset. We also focus on the linear (Head) classifier in our empirical evaluation in \cref{sec:predicting_dataset}. In the 
linear classifier, we find a matrix $W$ and biases $b$, to optimize
the cross-entropy between the labels and logits $Wv + b$, where $v$
denotes the feature space representation of the data. In the
simplified model, the rows of $W$ are replaced by the cluster means
and we do not include the bias term in the classification.

Now, applying Lemma \ref{lemma:per-example-LiRA} to the simplified model yields the following result.
\begin{restatable}[Per-example LiRA power-law]{theorem}{theoremLiRAPowerLaw}\label{theorem:LiRA-power-law}
    Fix a target example $(\bm{x},y)$. For the simplified model with arbitrary $C$ and infinitely many shadow models, the per-example LiRA vulnerability is given as
    \begin{multline}
        \log(\tpr_\lira(\bm{x}) - \fpr_\lira(\bm{x})) \\= -\frac{1}{2}\log S - \frac{1}{2}\Phi^{-1}(\fpr_\lira(\bm{x}))^2 + \log\frac{|\langle \bm{x}, \bm{x} - \bm{m}_{\bm{x}} \rangle|}{\sqrt{\bm{x}^T\Sigma \bm{x}}\sqrt{2\pi}} + \log(1 + \xi(S)),
    \end{multline}
    where $\bm{m}_{\bm{x}}$ is the true mean of class $y$ and $\xi(S)=O(1/\sqrt{S})$.
    For large $S$ we have
    \begin{equation}
        \log(\tpr_\lira(\bm{x}) - \fpr_\lira(\bm{x})) \approx -\frac{1}{2}\log S - \frac{1}{2}\Phi^{-1}(\fpr_\lira(\bm{x}))^2 + \log\frac{|\langle \bm{x}, \bm{x} - \bm{m}_{\bm{x}} \rangle|}{\sqrt{\bm{x}^T\Sigma \bm{x}}\sqrt{2\pi}}.
    \end{equation}
\end{restatable}
\begin{proof}
    See \cref{sec:proof-lira-power-law}.
\end{proof}
An immediate upper bound is obtained from Theorem \ref{theorem:LiRA-power-law} by the Cauchy-Schwarz inequality:
\begin{multline}
    \log(\tpr_\lira(\bm{x}) - \fpr_\lira(\bm{x})) \leq -\frac{1}{2}\log S - \frac{1}{2}\Phi^{-1}(\fpr_\lira(\bm{x}))^2 + \log\frac{||\bm{x} - \bm{m}_{\bm{x}}||}{\sqrt{\bm{x}^T\Sigma \bm{x}}\sqrt{2\pi}}\\ + \log(1 + \xi(S)).
\end{multline}
This implies that if $||\bm{x} - \bm{m}_{\bm{x}}||$ is bounded, then the worst-case vulnerability is also bounded.
Hence we can significantly reduce the MIA vulnerability of all examples in this non-DP setting by simply increasing the number of examples per class.

\begin{remark}\label{remark:from-scratch}
    By \cref{lemma:per-example-LiRA} and the proof of \cref{theorem:LiRA-power-law}, a necessary condition for the power-law is that $(\muin(\bf{x}) - \muout(\bf{x}))/\sigma(\bf{x})$ converges to zero at rate $O(1/S^\alpha)$ with $\alpha>0$.
    In our simplified model, this holds with $\alpha=1/2$.
    However, $\muin(\bf{x}) - \muout(\bf{x}) \rightarrow 0$ might not always be the case for larger neural networks trained from scratch. Furthermore, even if $(\muin(\bf{x}) - \muout(\bf{x}))/\sigma(\bf{x})$ converges to zero, it is not clear what the convergence rate would be.
\end{remark}

Now the following corollary extends the power-law to the average-case MIA vulnerability. 
We will also empirically validate this result in \cref{sec:predicting_dataset}.
\begin{restatable}[Average-case LiRA power-law]{corollary}{corAverageLiRA}\label{cor:average-LiRA}
    For the simplified model with arbitrary $C$, sufficiently large $S$ and infinitely many shadow models, we have
    \begin{equation}
        \log(\atpr_\lira - \afpr_\lira) \approx -\frac{1}{2}\log S - \frac{1}{2}\Phi^{-1}(\afpr_\lira)^2 + \log\lp \E_{(\bm{x}, y) \sim \sD}\left[\frac{|\langle \bm{x}, \bm{x} - \bm{m}_{\bm{x}} \rangle|}{\sqrt{\bm{x}^T\Sigma \bm{x}}\sqrt{2\pi}}\right] \rp.\label{eq:power-law}
    \end{equation}
\end{restatable}
\begin{proof}
    See \cref{sec:proof-average-lira}.
\end{proof}

\section{Empirical evaluation of MIA vulnerability and dataset properties}\label{sec:predicting_dataset}
In this section, we investigate how different properties of datasets (shots $S$ and number of classes $C$) affect the MIA vulnerability. Based on our observations, we propose a method to predict the vulnerability to MIA using these properties.

\subsection{Experimental setup}\label{sec:dataset_experimental_setup}

We focus on a image classification setting where we fine-tune pre-trained models on sensitive downstream datasets and assess the MIA vulnerability using LiRA and RMIA with $M = 256$ shadow/reference models. We base our experiments on a subset of the few-shot benchmark VTAB~\citep{zhai2019large} that achieves a test classification accuracy $> 80\%$ (see \cref{tab:used-datasets}).

We report results for fine-tuning a last layer classifier (Head) trained on top of a Vision Transformer ViT-Base-16~\citep[ViT-B;][]{dosovitskiy2020image}, pre-trained on ImageNet-21k~\citep{ILSVRC15}. The results for using ResNet-50~\citep[R-50;][]{kolesnikov2019big} as a backbone can be found in \cref{sec:additional_sec5}. We optimize the hyperparameters (batch size, learning rate and number of epochs) using the Optuna library~\citep{optuna_2019} with the Tree-structured Parzen Estimator~\citep[TPE;][]{bergstra_tpe_2011} sampler with 20 iterations (more details in \cref{sec:hyperparameter_tuning}). We provide the the code for reproducing the experiments in an open repository\footnote{
\url{https://github.com/DPBayes/impact-dataset-properties-MI-vulnerability-deep-TL}}.

\subsection{Experimental results}\label{sec:experimental_results}

\begin{wrapfigure}{r}{0.5\textwidth}
\vspace{-4mm}
\begin{minipage}{\linewidth}
    \includegraphics{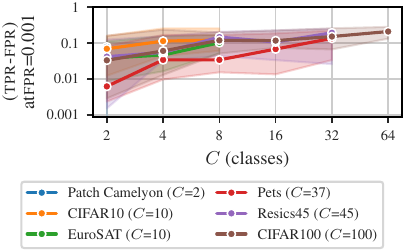}
    \caption{Small effect of number of classes $C$ (classes) on MIA vulnerability when attacking a  fine-tuned ViT-B Head. The solid line is median and the error bars the min/max bounds for the Clopper-Pearson CIs over 12 seeds ($S=32$).}\label{fig:mia_function_classes}
\end{minipage}
\end{wrapfigure}

Using the setting described above, we study how the number of classes and the number of shots affect the vulnerability ($\tpr$ at $\fpr$ as described in \cref{sec:background}) using LiRA. We make the following observations:

\begin{itemize}[leftmargin=*,noitemsep,topsep=0pt]
\setlength\itemsep{0pt} 
    \item A larger number of {\boldmath $\shots$} \textbf{(shots)} decrease the vulnerability in a power law relation as demonstrated in \cref{fig:fig1}. We provide tabular data and experiments using ResNet-50 in the Appendix (\cref{fig:mia_function_shots_additional,tab:vit-function-shots,tab:r-50-function-shots}).

    \item Contrary, a larger number of {\boldmath $\classes$} \textbf{(classes)} increases the vulnerability as demonstrated in \cref{fig:mia_function_classes} with tabular data and experiments using ResNet-50  in the Appendix (\cref{fig:mia_function_classes_additional,tab:vit-b-function-classes,tab:r-50-function-classes}). However, the trend w.r.t.\ $\classes$ is not as clear as with $\shots$.
\end{itemize}

\textbf{RMIA}
In \cref{fig:rmia-lira-comparison} we compare the vulnerability of the models to LiRA and RMIA as a function of the number of $\shots$ (shots) at $\fpr=0.1$. We observe the power-law for both attacks, but the RMIA is more unstable than LiRA (especially for lower $\fpr$). More results for RMIA are in \cref{fig:rmia_function_shots,fig:rmia_comparison_to_lira,fig:rmia-lira-comparison_additional} in the Appendix.

\begin{figure}[h]
    \centering
    \includegraphics[width=\linewidth]{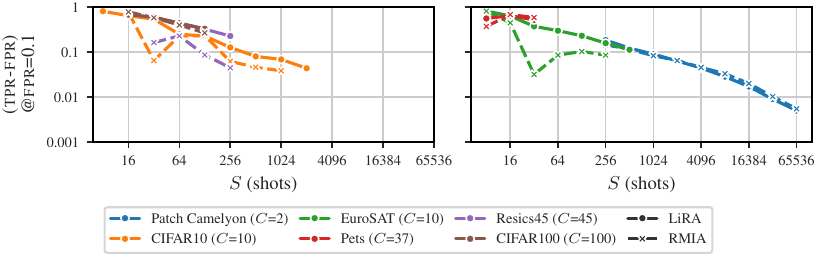}
    \caption{LiRA and RMIA vulnerability ($(\tpr-\fpr)$ at $\fpr=0.1$) as a function of shots ($\shots$) when attacking a ViT-B Head fine-tuned on different datasets. For better visibility, we split the datasets into two panels. We observe the power-law for both attacks, but the RMIA is more unstable than LiRA. The lines display the median over six seeds.}
    \label{fig:rmia-lira-comparison}
    \vspace{-4mm}
\end{figure}

\subsection{Model to predict dataset vulnerability}
\label{sec:model-to-predict-vulnerability}

The trends seen in \cref{fig:fig1} suggest the same power law relationship that we derived for the simplified model of membership inference in \cref{sec:explaining_trends}. We fit a linear regression model to predict $\log(\tpr - \fpr)$ for each $\fpr = 10^{-k}, k=1,\dots,5$ separately using the $\log \classes$ and $\log \shots$ as covariates with statsmodels~\citep{seabold2010statsmodels}. The general form of the model can be found in \cref{eq:mia_vul_dataset}, where $\beta_{\shots}, \beta_{\classes}$ and $\beta_{0}$ are the learnable regression parameters.

\begin{equation}\label{eq:mia_vul_dataset}
\log_{10}(\tpr-\fpr) = \beta_{\shots}\log_{10}(\shots)+\beta_{\classes}\log_{10}(\classes)+\beta_{0}
\end{equation}

 In \cref{sec:predicting_dataset_simpler}, we propose a variation of the regression model that predicts $\log_{10}(\tpr)$ instead of $\log_{10}(\tpr-\fpr)$ but this alternative model performs worse on our empirical data and predicts $\tpr < \fpr$ in the tail when $\shots$ is very large. 

We utilise MIA results of ViT-B (Head) (see \cref{tab:vit-function-shots}) as the training data. Based on the $R^2$ (coefficient of determination) score ($R^2=0.930$ for the model trained on $\fpr=0.001$ data), our model fits the data extremely well.  We provide further evidence for other $\fpr$ in \cref{fig:predict_mia_dataset_eval_additional} and \cref{tab:glm-coefs-diff} in the Appendix. \cref{fig:linear-regresion-params} shows the parameters of the prediction model fitted to the training data. For larger $\fpr$, the coefficient $\beta_S$ is around $-0.5$, as our theoretical analysis predicts.

\begin{figure}[t]
    \centering
    \includegraphics[width=\linewidth]{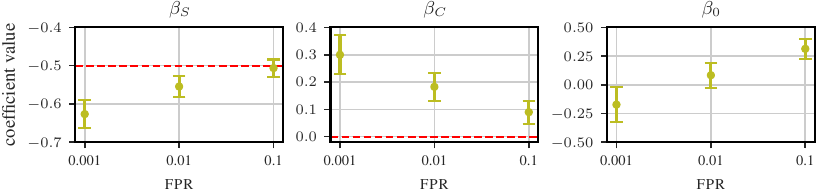}
    \caption{Coefficient values for different $\fpr$ when fitting a regression model based on \cref{eq:mia_vul_dataset} fitted on data from  ViT-B (Head) with LiRA (\cref{tab:vit-function-shots}). The error bars display the 95\% confidence intervals based  on Student's t-distribution. Theoretical values in the simplified model is shown by pink dotted lines ($\beta_S = 0.5$ and $\beta_C=0$).}
    \label{fig:linear-regresion-params}
    \vspace{-4mm}
\end{figure}

\paragraph{Prediction quality on other MIA target models} We analyze how the regression model trained on the ViT-B (Head) data generalizes to other target models. The main points are:
\begin{itemize}[leftmargin=*,topsep=-4pt]
    \item \textit{R-50 (Head):} \cref{fig:predict_mia_dataset_eval} shows that the regression model is robust to a change of the feature extractor, as it is able to predict the $\tpr$ for R-50 (Head) (test $R^2=0.790$). 
    \item \textit{R-50 (FiLM):} \cref{fig:eval_r50filmcarlini} shows that the prediction quality is good for R-50 (FiLM) models. These models are fine-tuned with parameter-efficient FiLM~\citep{perez2018film} layers (See \cref{sec:parameterization}). \citet{tobaben2023Efficacy} demonstrated that FiLM layers are a competitive alternative to training all parameters. We supplement the MIA results of \citet{tobaben2023Efficacy} with own FiLM training runs. Refer to \cref{tab:comparision_data} in the Appendix.
    \item \textit{From-Scratch-Training:} \citet{Carlini2022LiRA} provide limited results on from-scratch-training. To the best of our knowledge these are the only published LiRA results on image classification models. \cref{fig:eval_r50filmcarlini} displays that our prediction model underestimates the vulnerability of the from-scratch trained target models. We have identified two potential explanations for this: \begin{inlinelist}
        \item In from-scratch-training all weights of the model need to be trained from the sensitive data and thus potentially from-scratch-training could be more vulnerable than fine-tuning. 
        \item The strongest attack in \citet{Carlini2022LiRA} uses data augmentations to improve the performance. We are not using this optimization.
    \end{inlinelist}
    Additionally, as noted in \cref{remark:from-scratch}, our theoretical analysis is based on a simple model and the power-law might not occur at all for larger neural networks trained from scratch.
\end{itemize}

\begin{figure}[t]
		\centering
\begin{subfigure}[t]{0.465\textwidth}
        \includegraphics[width=\textwidth]{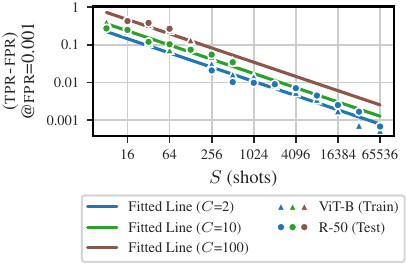}
        \caption{The dots show the median $\tpr$ for the train set (ViT-B; \cref{tab:vit-function-shots}) and the test set (R-50; \cref{tab:r-50-function-shots}) over six seeds (datasets: Patch Camelyon, EuroSAT and CIFAR100). The linear model is robust to changing the feature extractor from ViT-B to R-50.}\label{fig:predict_mia_dataset_eval}
    \end{subfigure}\hfill
    \begin{subfigure}[t]{0.49\textwidth}
        \includegraphics[width=\textwidth]{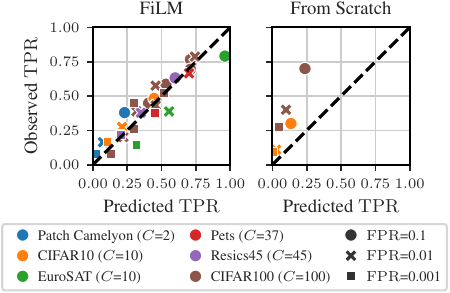}
        \caption{Regression model is robust to changing the fine-tuning method from Head to FiLM, but from scratch training seems to be more vulnerable than predicted.
        	\begin{inlinelist}
        		\item left: fine-tuned with FiLM (see \cref{tab:comparision_data})
        		\item right: trained from scratch. Data is from \citet{Carlini2022LiRA}
        	\end{inlinelist}. 
        }\label{fig:eval_r50filmcarlini}
    \end{subfigure}
\caption{Performance of the regression model based on \cref{eq:mia_vul_dataset} fitted on data from \cref{tab:vit-function-shots}.}\label{fig:regression_model_performance_large}
\vspace{-4mm}
\end{figure}

\begin{minipage}{0.49\linewidth}
    \includegraphics{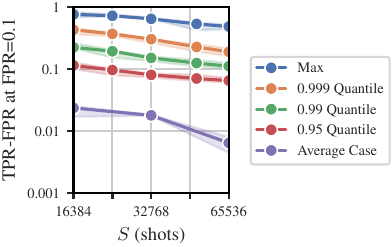}
    \captionof{figure}{Individual vulnerability for ViT-B (Head) when fine-tuning on Patch Camelyon. We observe a similar power-law relationship for individuals when looking at the quantiles but the max decreases slower. The Average Case shows the average LiRA vulnerability over the individuals also illustrated in the right panel of \cref{fig:rmia-lira-comparison}.}\label{fig:individual_mia}
\end{minipage}
\hfill
\begin{minipage}{0.49\linewidth}
        \captionof{table}{The minimum $S$ with $C$=$2$ predicted by the models in \cref{sec:model-to-predict-vulnerability,sec:individual_mia} to empirically match the DP bounds ($\delta$=$10^{-5}$) calculated through \citep{Kairouz2015Composition} in terms of $\tpr$ at $\fpr$.}
    \label{tab:bound_comparison}
\begin{adjustbox}{max width=\linewidth}
    \begin{tabular}{|c|rrr|r|}
    \toprule
    $\epsilon$ & \multicolumn{3}{c|}{ min $S$ in average case} & min $S$ in \\
    & & & &worst-case \\
     & $\fpr$=$0.1$ & $\fpr$=$0.01$ & $\fpr$=$0.001$ & $\fpr$=$0.1$ \\
    \midrule
    $0.25$ & $5\,400$ & $69\,000$ & $320\,000$ & $5.5\times 10^9$ \\
    $0.50$ & $1\,100$ & $16\,000$ & $88\,000$ & $2.6\times 10^8$ \\
    $0.75$ & $360$ & $5\,900$ & $38\,000$ & $3.5\times 10^7$ \\
    $1.00$ & $160$ & $2\,700$ & $19\,000$ & $7.0\times 10^6$ \\
    \bottomrule
    \end{tabular}
\end{adjustbox}
    \vspace{5\baselineskip}
\end{minipage}

\begin{figure}[ht]
    \centering
    \includegraphics[width=0.48\linewidth]{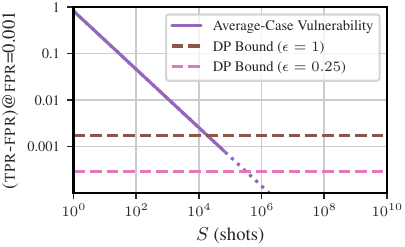}
    \includegraphics[width=0.48\linewidth]{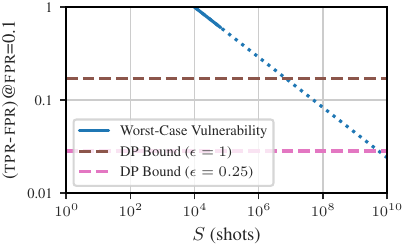}
    \caption{Illustration of the extrapolation for the average-case at $\fpr=10^{-3}$ (left) and worst-case at $\fpr=10^{-1}$ (right).}
    \label{fig:powerlaw_extrapolation_illustration}
\end{figure}

\subsection{Individual MIA vulnerability}\label{sec:individual_mia}
 In order to assess the per-sample MIA vulnerability, we run the experiment with $257$ models with each of them once acting as a target and as a shadow model otherwise, compute the $\tpr$ at $\fpr$ for every sample separately. In \cref{fig:individual_mia}, we display the individual vulnerability as a function of $S$ (shots) for Patch Camelyon and compare it with corresponding average case vulnerability from \cref{fig:rmia-lira-comparison}. The plot shows the maximal vulnerability of all samples and different quantiles over six seeds. The solid line is the median and the errorbars display the min and max over seeds. The quantiles are more robust to extreme outliers and show decreasing trends already at much lower $S$ than the maximum vulnerability.

When fitting the model in \cref{eq:mia_vul_dataset} we observe that the coefficients $\beta_{\shots}$ that model the relationship between vulnerability and examples per class for the quantiles are $-0.5603$, $-0.5688$ and $-0.4796$ which is close to the theoretical value of $-0.5$ derived in the theoretical analysis in \cref{sec:simplified-model}. However, the maximum vulnerability decreases with a lower slope of $-0.2695$ which is considerably smaller. With larger $S$ the slope of the max vulnerability increases, e.g., with $S\ge32768$ the coefficient is $-0.3478$ suggesting that higher $S$ is required for the most vulnerable points.

\subsection{Comparison between empirical models and universal DP bounds}\label{sec:bound_comparision}

While the practical evaluation through MIAs is statistical and does not provide universal formal guarantees like DP, the power-law can aid understanding about how the practical vulnerability to MIA behaves in a more realistic threat model when the examples per class increase.

Using the translation between the $(\tpr,\fpr)$ to $(\epsilon,\delta)$-DP guarantees proposed by \citet{Kairouz2015Composition}, we compute the minimum $S$ predicted by our empirical models such that the predicted $\tpr$ matches the theoretical bound at target DP level. We illustrate this in \cref{fig:powerlaw_extrapolation_illustration}. (See \cref{app:bound_comparison} for a more detailed description.)
\cref{tab:bound_comparison} shows the resulting lower bounds of $S$ for various DP levels and values of $\fpr$. While our empirical observations do not provide any formal guarantees, the comparison serves as an illustration to better understand the different requirements of the average and individual vulnerability.  
We can see that both for the average and the worst case, obtaining a low $\fpr$ would require a large amount of samples per class in order to match the $\tpr$ for meaningfully strong DP bounds. 

\section{Discussion}\label{sec:discussion}
Trying to bridge empirical MIA vulnerability and formal DP guarantees is not an easy task because of different threat models and different nature of bounds (statistical vs.\ universal).
While we were able to show both theoretically (\Cref{sec:explaining_trends}) and empirically (\Cref{sec:predicting_dataset}) that having more examples per class provides protection against MIA in fine-tuned neural networks, the numbers required for significant protection (\Cref{sec:bound_comparision}) limit the practical utility of this observation.

Using a level of formal DP bounds that provide meaningful protection as a yardstick, at least tens of thousands of examples per class are needed for every class even at $\fpr=0.001$.
The $\epsilon$ values used here are formal upper bounds and must not be compared to actual DP deep learning with comparable privacy budgets, as the latter would be far less vulnerable.
At this number of samples per class and a good pre-trained model, the impact of DP training is often negligible.
This stresses the importance of formal guarantees like DP for privacy protection.

Our formal analysis focuses on a setting that can be linked to fine-tuning. As shown in \cref{fig:eval_r50filmcarlini}, from-scratch training likely has higher vulnerability. Formally analyzing more models is an interesting area for future research.

Since our MIA evaluation assumes that only a target point is known to the adversary and the rest of the dataset is random, this stochasticity likely introduces some protection.
Hence the power-law may be invalidated under stronger MIA settings where the adversary has access to other points in the private dataset~\citep{Bai2025OtherMIAs}.

Our experiments show that there is a difference between classes in terms of vulnerability and an interesting direction for future work is to understand the properties of classes that influence this vulnerability, e.g., variability within or between classes or their separability. Another direction for future work is understanding the impact of pre-training data on the vulnerability of fine-tuning data.

\paragraph{Broader Impact}
Our work systematically studies factors influencing privacy risk of trained ML models. This has significant positive impact on users training ML models on personal data by allowing them to understand and limit the risks.

\paragraph{Limitations}
We mostly consider LiRA in our paper which is optimal for our simplified model of membership inference (\cref{sec:simplified-model}), but for the transfer-learning experiments (\cref{sec:predicting_dataset}) there might be stronger attacks in the future. Furthermore, our simplified model assumes well-behaved underlying distributions, meaning that the data is normally distributed around the class centers. We leave the analysis of other data distributions (e.g., heavy-tailed distributions) to future work. Formal bounds on MIA vulnerability would require something like DP. In addition, both our theoretical and empirical analysis focus on deep transfer learning using fine-tuning. Models trained from scratch are likely to be more vulnerable. Concurrently, a related paper~\citep{Hayes2025LLMMIA} has studied the relation between LiRA vulnerability and dataset size when attacking LLMs during the pre-training phase. Extending this analysis is promising but computationally very expensive.

We provide the code in an open repository\footnote{
\url{https://github.com/DPBayes/impact-dataset-properties-MI-vulnerability-deep-TL}}. All used pre-trained models and datasets are publicly available.

\section*{Acknowledgments}
This work was supported by the Research Council of Finland (Flagship programme: Finnish
Center for Artificial Intelligence, FCAI, Grant 356499 and Grant 359111), the Strategic Research Council
at the Research Council of Finland (Grant 358247), the European Union (Project
101070617) as well as JSPS KAKENHI Grant Number 25KJ1515. Views and opinions expressed are however those of the author(s) only and do not
necessarily reflect those of the European Union or the European Commission. Neither the
European Union nor the granting authority can be held responsible for them.
he authors wish to acknowledge CSC – IT Center for Science, Finland, for computational resources.
We thank Mikko A. Heikkilä and Ossi Räisä for helpful comments and suggestions and John F. Bronskill for helpful discussions regarding few-shot learning.

\bibliography{main}
\bibliographystyle{abbrvnat}



\clearpage 

\appendix
\addtocontents{toc}{\protect\setcounter{tocdepth}{2}}
\hrule height 4pt
  \vskip 0.25in
  
\begin{center}
    \textbf{\huge Supplementary Material}
\end{center}

 \vskip 0.1in
  \hrule height 1pt

\tableofcontents
\clearpage

\setcounter{figure}{0}                       
\renewcommand\thefigure{A.\arabic{figure}}
\setcounter{table}{0}
\renewcommand{\thetable}{A\arabic{table}}
\setcounter{equation}{0}
\renewcommand{\theequation}{A\arabic{equation}}
\setcounter{lemma}{0}
\renewcommand{\thelemma}{A\arabic{lemma}}
\setcounter{corollary}{0}
\renewcommand{\thecorollary}{A\arabic{corollary}}

\section{Details of \cref{sec:explaining_trends}}\label{sec:theory_details}

\subsection{Formulating LiRA}\label{sec:formulate-lira-rmia}
Let $\mathcal{M}$ be our target model and $\ell(\mathcal{M}(\bm{x}), y)$ be the loss of the model on a target example $(\bm{x}, y)$.
The goal of MIA is to determine whether $(\bm{x}, y) \in \dtarget$.
This can be formulated as a hypothesis test:
\begin{align}
    &H_0: (\bm{x}, y) \notin \dtarget \label{eq:h_0} \\
    &H_1: (\bm{x}, y) \in \dtarget. \label{eq:h_1}
\end{align}
Following \citep{Carlini2022LiRA}, we formulate the Likelihood Ratio Attack (LiRA).
LiRA exploits the difference of losses on the target model under $H_0$ and $H_1$.
To model the IN/OUT loss distributions with few shadow models, LiRA employs a parametric modelling.
Particularly, LiRA models $\tx$ by a normal distribution.
That is, the hypothesis test formulated above can be rewritten as
\begin{align}
    &H_0^\prime: \tx \sim \mathcal{N}(\hmuout, \hsigout^2) \\
    &H_1^\prime: \tx \sim \mathcal{N}(\hmuin, \hsigin^2).
\end{align}
The likelihood ratio is now
\begin{equation}
    \mathrm{LR}(\bm{x}) = \frac{\mathcal{N}(\tx; \hmuin, \hsigin^2)}{\mathcal{N}(\tx; \hmuout, \hsigout^2)}.
\end{equation}
LiRA rejects $H_0^\prime$ if and only if
\begin{equation}
    \LR(\bm{x}) \geq \tau,
\end{equation}
concluding that $H_1^\prime$ is true, i.e., identifying the membership of $(\bm{x}, y)$.
Thus, the true positive rate of this hypothesis test given as
\begin{equation}
    \tpr_\lira(\bm{x}) = \Pr_{\dtarget\sim\sD^{|\data|}, \phi^M} \lp \LR(\bm{x}) \geq \tau \mid (\bm{x}, y) \in \dtarget\rp,
\end{equation}
where $\phi^M$ denotes the randomness in the shadow set sampling and shadow model training.

\subsection {On the assumption of shared scale}\label{sec:shared-scale}
In \cref{sec:explaining_trends} we assumed that for LiRA $\sigin=\sigout$ and $\hsigin=\hsigout$.
Using the simplified model formulated in \cref{sec:simplified-model}, we show that for large enough number $S$ of examples per class these assumptions are reasonable.

Let $\dtarget = \{(\bm{x}_{j,1}, j),...,(\bm{x}_{j,S}, j)\}_{j=1}^C$. Then the IN/OUT LiRA scores are given as
\begin{align}
    s_{y}^\mathrm{(in)} &= \langle \bm{x}, \frac{1}{S} \lp \sum_{i=1}^{S-1} \bm{x}_{y, i} + \bm{x}\rp \rangle = \langle \bm{x}, \frac{1}{S} \sum_{i=1}^S \bm{x}_{y, i} \rangle + \langle \bm{x}, \frac{1}{S}(\bm{x} - \bm{x}_{y, S}) \rangle \\
    s_{y}^\mathrm{(out)} &= \langle \bm{x}, \frac{1}{S} \sum_{i=1}^S \bm{x}_{y, i} \rangle.
\end{align}
Since for the simplified model scores follow Gaussian distributions, $\sigin=\hsigin$ and $\sigout = \hsigout$. It follows that
\begin{align}
    \sigin^2 &= \hsigin^2 = \Var(s_{y}^\mathrm{(in)}) = \frac{1}{S}\Var(\langle \bm{x}, \bm{x}_{y, i} \rangle) - \frac{1}{S^2}\Var(\langle \bm{x}, \bm{x}_{y, i} \rangle) = \frac{1}{S}\lp1 - \frac{1}{S}\rp \bm{x}^T\Sigma \bm{x} \\
    \sigout^2 &= \hsigout^2 = \Var(s_{y}^\mathrm{(out)}) = \frac{1}{S}\Var(\langle \bm{x}, \bm{x}_{y, i} \rangle) = \frac{1}{S}\bm{x}^T\Sigma \bm{x}.
\end{align}
Thus, the differences $\sigin-\sigout$ and $\hsigin-\hsigout$ are negligible for large $S$.

\subsection{Proof of \cref{lemma:per-example-LiRA}}\label{sec:proof-per-example-LiRA}
\lemmaPerExampleLiRA*
\begin{proof}
    We abuse notations by denoting $\muin$ to refer to $\muin(\bm{x})$ and similarly for other statistics.
    We have
    \begin{align}
        \log\frac{\mathcal{N}(\tx; \hmuin, \hsig^2)}{\mathcal{N}(\tx; \hmuout, \hsig^2)} \geq \log\tau \\
        -\frac{1}{2}\lp \frac{\tx - \hmuin}{\hsig} \rp^2 + \frac{1}{2}\lp \frac{\tx - \hmuout}{\hsig} \rp^2 \geq \log\tau \\
        \frac{1}{2\hsig^2}(2\tx\hmuin - \hmuin^2 - 2\tx\hmuout + \hmuout^2) \geq \log\tau \\
        \frac{1}{2\hsig^2}(\hmuin - \hmuout)(2\tx - \hmuin - \hmuout) \geq \log\tau \\
        \begin{cases}
            \tx \geq \frac{\hsig^2\log\tau}{\hmuin - \hmuout} + \frac{\hmuin + \hmuout}{2} &\text{ if } \hmuin > \hmuout\\
            \tx \leq \frac{\hsig^2\log\tau}{\hmuin - \hmuout} + \frac{\hmuin + \hmuout}{2} &\text{ if } \hmuin < \hmuout.
        \end{cases}
    \end{align}
    Then if $\hmuin > \hmuout$, in the limit of infinitely many shadow models
    \begin{align}
        \fpr_\lira(\bm{x}) &= \Pr_Z\lp \muout + \sigma Z \geq \frac{\hsig^2\log\tau}{\hmuin - \hmuout} + \frac{\hmuin + \hmuout}{2} \rp \\
        &= \Pr_Z\lp Z \geq \frac{\hsig^2\log\tau}{\sigma(\hmuin - \hmuout)} + \frac{\hmuin + \hmuout}{2\sigma} - \frac{\muout}{\sigma} \rp \\
        &= 1 - F_Z\lp \frac{\hsig^2\log\tau}{\sigma(\hmuin - \hmuout)} + \frac{\hmuin + \hmuout}{2\sigma} - \frac{\muout}{\sigma} \rp,
    \end{align}
    and if $\hmuin < \hmuout$, similarly,
    \begin{align}
        \fpr_\lira(\bm{x}) &= \Pr_Z\lp \muout + \sigma Z \leq \frac{\hsig^2\log\tau}{\hmuin - \hmuout} + \frac{\hmuin + \hmuout}{2} \rp \\
        &= F_Z\lp \frac{\hsig^2\log\tau}{\sigma(\hmuin - \hmuout)} + \frac{\hmuin + \hmuout}{2\sigma} - \frac{\muout}{\sigma} \rp.
    \end{align}
    Thus
    \begin{equation} \label{eq:lira-fpr-tau}
        \frac{\hsig^2\log\tau}{\sigma(\hmuin - \hmuout)} + \frac{\hmuin + \hmuout}{2\sigma} - \frac{\muout}{\sigma} = \begin{cases}
            F_Z^{-1}(1 - \fpr_\lira(\bm{x})) &\text{ if } \hmuin > \hmuout \\
            F_Z^{-1}(\fpr_\lira(\bm{x})) &\text{ if } \hmuin < \hmuout.
        \end{cases}
    \end{equation}
    It follows that if $\hmuin > \hmuout$,
    \begin{align}
        \tpr_\lira(\bm{x}) &= \Pr_Z\lp \muin + \sigma Z \geq \frac{\hsig^2\log\tau}{\hmuin - \hmuout} + \frac{\hmuin + \hmuout}{2} \rp \\
        &= \Pr_Z\lp Z \geq \frac{\hsig^2\log\tau}{\sigma(\hmuin - \hmuout)} + \frac{\hmuin + \hmuout}{2\sigma} - \frac{\muin}{\sigma} \rp \\
        &= 1 - F_Z\lp F_Z^{-1}(1 - \fpr_\lira(\bm{x})) - \frac{\muin - \muout}{\sigma} \rp.
    \end{align}
    If $\hmuin < \hmuout$, then
    \begin{align}
        \tpr_\lira(\bm{x}) &= \Pr_Z\lp \muin + \sigma Z \leq \frac{\hsig^2\log\tau}{\hmuin - \hmuout} + \frac{\hmuin + \hmuout}{2} \rp \\
        &= \Pr_Z\lp Z \leq \frac{\hsig^2\log\tau}{\sigma(\hmuin - \hmuout)} + \frac{\hmuin + \hmuout}{2\sigma} - \frac{\muin}{\sigma} \rp \\
        &= F_Z\lp F_Z^{-1}(\fpr_\lira(\bm{x})) - \frac{\muin - \muout}{\sigma} \rp.
    \end{align}
\end{proof}

\subsection{Offline LiRA}\label{sec:offline-lira}

\citet{Carlini2022LiRA} proposes offline LiRA that only trains OUT shadow models for computational efficiency. Instead of likelihood ratio, the score for offline LiRA is given as
\begin{equation}
    \Lambda(\bm{x}) = \Pr_Z(\hmuout(\bm{x}) + \hsigout(\bm{x}) Z \leq t_{\bm{x}}),
\end{equation}
where $Z$ is a location-scale distribution with standard location and unit scale. Then $\tpr$ of offline LiRA becomes
\begin{equation}
    \tpr_{\textrm{off}\lira}(\bm{x}) = \Pr_{\dtarget\sim\sD^{|\data|}, \phi^M} \lp \Lambda(\bm{x}) \geq \tau \mid (\bm{x}, y) \in \dtarget\rp,
\end{equation}
where $\gamma$ is a tunable parameter.
$\fpr$ is also given similarly. Below we show a result for offline LiRA similar to \cref{lemma:per-example-LiRA}, but under a slightly more strict assumption that an attacker accurately estimates the IN/OUT score distributions. In the following, as for online LiRA, we assume $\sigin=\sigout$.

\begin{lemma}[Per-example offline LiRA vulnerability]\label{lemma:per-example-offline-LiRA}
    Suppose that $\tx$ follows the normal distribution with means $\muin(\bm{x}), \muout(\bm{x})$ and standard deviation $\sigma(\bm{x})$.
    Assume that an attacker has access to the underlying distribution $\sD$.
    Then for a large enough number of examples per class and infinitely many shadow models, the offline LiRA vulnerability of a fixed target example is
    \begin{equation}
        \tpr_\offlira(\bm{x}) = \Phi\lp \Phi^{-1} \lp \fpr_\offlira\rp + \frac{\muin(\bm{x}) - \muout(\bm{x})}{\sigma(\bm{x})}\rp
    \end{equation}
    where $\Phi$ is the standard normal cdf.
\end{lemma}
\begin{proof}
    For infinitely many shadow models, we have
    \begin{equation}
        \fpr_{\offlira}(\bm{x}) = \Pr_{\dtarget\sim\sD^{|\data|}} \lp \Lambda(\bm{x}) \geq \tau \mid (\bm{x}, y) \notin \dtarget\rp,
    \end{equation}
    and the score is now
    \begin{equation}
        \Lambda(\bm{x}) = \Pr_\eta(\muout(\bm{x}) + \sigma(\bm{x}) \eta \leq t_{\bm{x}}),
    \end{equation}
    where $\eta \sim \mathcal{N}(0, 1)$.
    When $(\bm{x}, y) \notin \dtarget$, $t_{\bm{x}} = \muout(\bm{x}) + \sigma(\bm{x}) Z$. Thus we have
    \begin{align}
        \fpr_{\offlira} &= \Pr_Z \lp \Lambda(\bm{x}) \geq \gamma \rp \\
        &= \Pr_Z \lp \Pr_\eta(\muout(\bm{x}) + \sigma(\bm{x}) \eta \leq \muout(\bm{x}) + \sigma(\bm{x}) Z) \geq \gamma \rp \\
        &= \Pr_Z \lp \Pr_\eta(\eta \leq  Z) \geq \gamma \rp \\
        &= \Pr_Z (\Phi(Z) \geq \gamma) \\
        &= 1 - \gamma.
    \end{align}
    On the other hand, when $(\bm{x}, y) \in \dtarget$, $t_{\bm{x}} = \muin(\bm{x}) + \sigma(\bm{x}) Z$. Thus we obtain
    \begin{align}
        \tpr_\offlira &= \Pr_Z \lp \Pr_\eta(\muout(\bm{x}) + \sigma(\bm{x}) \eta \leq \muin(\bm{x}) + \sigma(\bm{x}) Z) \geq \gamma \rp \\
        &= \Pr_Z \lp \Pr_\eta \lp\eta \leq \frac{\muin(\bm{x}) - \muout(\bm{x})}{\sigma(\bm{x})} +  Z\rp\geq 1 - \fpr_\offlira \rp \\
        &= \Pr_Z \lp \Pr_\eta \lp\eta \leq - \frac{\muin(\bm{x}) - \muout(\bm{x})}{\sigma(\bm{x})} -  Z\rp\leq \fpr_\offlira \rp \\
        &= \Pr_Z \lp - \frac{\muin(\bm{x}) - \muout(\bm{x})}{\sigma(\bm{x})} -  Z \leq \Phi^{-1}(\fpr_\offlira) \rp \\
        &= \Phi\lp \Phi^{-1} \lp \fpr_\offlira\rp + \frac{\muin(\bm{x}) - \muout(\bm{x})}{\sigma(\bm{x})}\rp.
    \end{align}
\end{proof}

Consequently, the power-law also holds for offline LiRA in the simplified model:

\begin{corollary}[Per-example offline LiRA power-law]\label{cor:offline-LiRA-power-law}
    Fix a target example $(\bm{x},y)$. For the simplified model with arbitrary $C$ and infinitely many shadow models, the per-example offline LiRA vulnerability is given as
    \begin{multline}
        \log(\tpr_\offlira(\bm{x}) - \fpr_\offlira(\bm{x})) \\= -\frac{1}{2}\log S - \frac{1}{2}\Phi^{-1}(\fpr_\offlira(\bm{x}))^2 + \log\frac{|\langle \bm{x}, \bm{x} - \bm{m}_{\bm{x}} \rangle|}{\sqrt{\bm{x}^T\Sigma \bm{x}}\sqrt{2\pi}} + \log(1 + \xi(S))
    \end{multline}
    where $m_{\bm{x}}$ is the true mean of class $y$ and $\xi(S)=O(1/\sqrt{S})$.
    For large $S$ we have
    \begin{equation}
        \log(\tpr_\offlira(\bm{x}) - \fpr_\offlira(\bm{x})) \approx -\frac{1}{2}\log S - \frac{1}{2}\Phi^{-1}(\fpr_\lira(\bm{x}))^2 + \log\frac{|\langle \bm{x}, \bm{x} - \bm{m}_{\bm{x}} \rangle|}{\sqrt{\bm{x}^T\Sigma \bm{x}}\sqrt{2\pi}}.
    \end{equation}
\end{corollary}

\subsection{Relaxing the assumption of \cref{lemma:per-example-LiRA}}\label{sec:relax-assumption}
In \cref{lemma:per-example-LiRA} we assume that an attacker has access to the true underlying distribution. 
However, in real-world settings the precise underlying distribution may not be available for an attacker. In the following, noting that the \cref{eq:lira-px-tpr} mainly relies on the true location parameters $\muin(\bm{x}),\muout(\bm{x})$ and scale parameter $\sigma(\bm{x})$, we relax this assumption of distribution availability so that the attacker only needs an approximated underlying distribution for the optimal LiRA that achieves the performance in \cref{lemma:per-example-LiRA}.

First, notice that if we completely drop this assumption so that an attacker trains shadow models with an arbitrary underlying distribution, then we may not be able to choose a desired $\fpr_\lira(\bm{x})$.
From \cref{eq:lira-fpr-tau} we have
\begin{align}
    \frac{\hsig^2\log\tau}{\sigma(\hmuin - \hmuout)} + \frac{\hmuin + \hmuout}{2\sigma} - \frac{\muout}{\sigma} = \begin{cases}
        F_Z^{-1}(1 - \fpr_\lira(\bm{x})) &\text{ if } \hmuin > \hmuout \\
        F_Z^{-1}(\fpr_\lira(\bm{x})) &\text{ if } \hmuin < \hmuout
    \end{cases} \\
    \log\tau = \begin{cases}
        \frac{\hmuin - \hmuout}{\hsig^2}\lp \sigma F_Z^{-1}(1 - \fpr_\lira(\bm{x})) - \frac{\hmuin + \hmuout}{2} + \muout \rp &\text{ if } \hmuin > \hmuout \\
        \frac{\hmuin - \hmuout}{\hsig^2}\lp \sigma F_Z^{-1}(\fpr_\lira(\bm{x})) - \frac{\hmuin + \hmuout}{2} + \muout \rp &\text{ if } \hmuin < \hmuout.
    \end{cases}
\end{align}
Since it does not make sense to choose a rejection region of the likelihood ratio test such that $\tau < 1$, we assume $\tau \geq 1$. Then we have
\begin{equation}
    \begin{cases}
        \frac{\hmuin - \hmuout}{\hsig^2}\lp \sigma F_Z^{-1}(1 - \fpr_\lira(\bm{x})) - \frac{\hmuin + \hmuout}{2} + \muout \rp \geq 0 &\quad \text{ if }\ \hmuin > \hmuout \\
        \frac{\hmuin - \hmuout}{\hsig^2}\lp \sigma F_Z^{-1}(\fpr_\lira(\bm{x})) - \frac{\hmuin + \hmuout}{2} + \muout \rp \geq 0 &\quad \text{ if }\ \hmuin < \hmuout.
    \end{cases}
\end{equation}
Therefore, a sufficient condition about attacker's knowledge on the underlying distribution for \cref{lemma:per-example-LiRA} to hold is
\begin{equation}
    \begin{cases}
        \sigma F_Z^{-1}(1 - \fpr_\lira(\bm{x})) - \frac{\hmuin + \hmuout}{2} + \muout \geq 0 &\quad \text{ if }\ \hmuin > \hmuout \\
        \sigma F_Z^{-1}(\fpr_\lira(\bm{x})) - \frac{\hmuin + \hmuout}{2} + \muout \leq 0 &\quad \text{ if }\ \hmuin < \hmuout.
    \end{cases}
\end{equation}
Rearranging the terms yields
\begin{equation}
    \begin{cases}
        \frac{\hmuin + \hmuout}{2} - \muout \leq \sigma F_Z^{-1}(1 - \fpr_\lira(\bm{x})) &\quad \text{ if }\ \hmuin > \hmuout \\
        \frac{\hmuin + \hmuout}{2} - \muout \geq \sigma F_Z^{-1}(\fpr_\lira(\bm{x})) &\quad \text{ if }\ \hmuin < \hmuout.
    \end{cases}
\end{equation}
For example, if the estimated mean $\hmuout(<\hmuin)$ is too large compared to the true parameter $\muout$, the left hand side for the case $\hmuin > \hmuout$ becomes very large, thereby forcing us to choose sufficiently small $\fpr_\lira(\bm{x})$. Similarly, if $\hmuout(>\hmuin)$ is much smaller than $\muout$, then the range of possible values of $\fpr_\lira(\bm{x})$ will be limited. We summarise this discussion in the following:
\begin{lemma}[\cref{lemma:per-example-LiRA} with relaxed assumptions]
        Suppose that the true IN/OUT distributions of $\tx$ are of a location-scale family with locations $\muin(\bm{x}),\muout(\bm{x})$ and a shared scale $\sigma(\bm{x})$ such that the distributions have finite first and second moments. Assume that LiRA models $\tx$ by $\mathcal{N}(\hmuin(\bm{x}), \hsig(\bm{x})^2)$ and $\mathcal{N}(\hmuout(\bm{x}),\hsig(\bm{x})^2)$, and that in the limit of infinitely many shadow models, estimated parameters $\hmuin(\bm{x}), \hmuout(\bm{x})$ and $\hsig(\bm{x})$ satisfy the following:
    \begin{equation}
    \begin{cases}
        \frac{\hmuin + \hmuout}{2} - \muout \leq \sigma F_Z^{-1}(1 - \fpr_\lira(\bm{x})) &\quad \text{ if }\ \hmuin > \hmuout \\
        \frac{\hmuin + \hmuout}{2} - \muout \geq \sigma F_Z^{-1}(\fpr_\lira(\bm{x})) &\quad \text{ if }\ \hmuin < \hmuout.
    \end{cases}
    \end{equation}
    Then in the limit of infinitely many shadow models, the LiRA vulnerability of a fixed target example $(\bm{x}, y)$ is
    \begin{equation}
        \tpr_\lira(\bm{x}) = \begin{cases}
        1 - F_Z\lp F_Z^{-1}(1 - \fpr_\lira(\bm{x})) - \frac{\muin(\bm{x}) - \muout(\bm{x})}{\sigma(\bm{x})}\rp &\text{ if }\ \hmuin(\bm{x}) > \hmuout(\bm{x}) \\
        F_Z\lp F_Z^{-1}(\fpr_\lira(\bm{x})) - \frac{\muin(\bm{x}) - \muout(\bm{x})}{\sigma(\bm{x})}\rp &\text{ if }\ \hmuin(\bm{x}) < \hmuout(\bm{x}),
        \end{cases}
    \end{equation}
    where $F_Z$ is the cdf of $t$ with the standard location and unit scale, assuming that the inverse of $F_Z$ exists.
\end{lemma}

\subsection{Proof of \cref{theorem:LiRA-power-law}}\label{sec:proof-lira-power-law}
\theoremLiRAPowerLaw*   
\begin{proof}
    Let $\dtarget = \{(\bm{x}_{j,1}, j),...,(\bm{x}_{j,S}, j)\}_{j=1}^C$. Then the LiRA score of the target $(\bm{x}, y)$ is
    \begin{align}
        s_{y}^\mathrm{(in)} &= \langle \bm{x}, \frac{1}{S} \lp \sum_{i=1}^{S-1} \bm{x}_{y, i} + \bm{x}\rp \rangle = \langle \bm{x}, \frac{1}{S} \sum_{i=1}^S \bm{x}_{y, i} \rangle + \langle \bm{x}, \frac{1}{S}(\bm{x} - \bm{x}_{y, S}) \rangle \\
        s_{y}^\mathrm{(out)} &= \langle \bm{x}, \frac{1}{S} \sum_{i=1}^S \bm{x}_{y, i} \rangle,
    \end{align}
    respectively, when $(\bm{x}, y) \in \dtarget$ and when $(\bm{x}, y) \notin \dtarget$.
    Thus we obtain
    \begin{align}
        \muin - \muout &= \E[s_{y}^\mathrm{(in)} - s_{y}^\mathrm{(out)}] = \frac{1}{S}\langle \bm{x}, \bm{x} - \bm{m}_{\bm{x}}\rangle \\
        \sigma^2 &= \Var(s_{y}^\mathrm{(out)}) = \frac{1}{S} \Var(\langle \bm{x}, \bm{x}_{y, i}\rangle) = \frac{1}{S}\bm{x}^T\Sigma \bm{x}.
    \end{align}
    Noting that the LiRA score follows a normal distribution, by Lemma \ref{lemma:per-example-LiRA} we have
    \begin{align}
        \tpr_\lira(\bm{x}) &= \Phi\lp \Phi^{-1}(\fpr_\lira(\bm{x})) + \frac{|\langle \bm{x}, \bm{x}- \bm{m}_{\bm{x}}\rangle|}{\sqrt{S}\sqrt{\bm{x}^T\Sigma \bm{x}}} \rp,
    \end{align}
    where $\Phi$ is the cdf of the standard normal distribution. 
    This completes the first half of the theorem.

    Now let $\phi(u)$ denote the pdf of the standard normal distribution, and let
    \begin{equation}
        r =  \frac{|\langle \bm{x}, \bm{x}- \bm{m}_{\bm{x}}\rangle|}{\sqrt{S}\sqrt{\bm{x}^T\Sigma \bm{x}}}.
    \end{equation}
    Using Taylor expansion of $\Phi( \Phi^{-1}(\fpr_\lira(\bm{x})) + r)$ around $r=0$, we have
    \begin{align}
        \tpr_\lira(\bm{x}) &= \sum_{k=0}^\infty \frac{\Phi^{(k)}(\Phi^{-1}(\fpr_\lira(\bm{x}))}{k!}r^k \\
        &= \fpr_\lira(\bm{x}) + \sum_{k=1}^\infty \frac{\Phi^{(k)}(\Phi^{-1}(\fpr_\lira(\bm{x}))}{k!}r^k \\
        &= \fpr_\lira(\bm{x}) + \sum_{k=1}^\infty \frac{\phi^{(k-1)}(\Phi^{-1}(\fpr_\lira(\bm{x}))}{k!}r^k \\
        &= \fpr_\lira(\bm{x}) + \sum_{k=1}^\infty \frac{(-1)^{k-1}\He_{k-1}(\Phi^{-1}(\fpr_\lira(\bm{x})))\phi(\Phi^{-1}(\fpr_\lira(\bm{x}))}{k!}r^k \\
        &= \fpr_\lira(\bm{x}) + r \phi(\Phi^{-1}(\fpr_\lira(\bm{x})) \sum_{k=0}^\infty \frac{(-1)^k\He_k(\Phi^{-1}(\fpr_\lira(\bm{x})))}{(k+1)!}r^k,
    \end{align}
    where $\He$ denotes Hermite polynomials.
    It follows that
    \begin{align}
        &\log(\tpr_\lira(\bm{x}) - \fpr_\lira(\bm{x})) \\
        =& \log r + \log \phi(\Phi^{-1}(\fpr_\lira(\bm{x})) + \log\lp \sum_{k=0}^\infty \frac{(-1)^k\He_k(\Phi^{-1}(\fpr_\lira(\bm{x})))}{(k+1)!}r^k \rp \\
        =& -\frac{1}{2}\log S - \frac{1}{2}\Phi^{-1}(\fpr_\lira(\bm{x}))^2 + \log\frac{|\langle \bm{x}, \bm{x} - \bm{m}_{\bm{x}} \rangle|}{\sqrt{\bm{x}^T\Sigma \bm{x}}\sqrt{2\pi}} \\
        &+ \log\lp \sum_{k=0}^\infty \frac{(-1)^k \He_k(\Phi^{-1}(\fpr_\lira(\bm{x})))}{(k+1)!} \lp \frac{|\langle \bm{x}, \bm{x} - \bm{m}_{\bm{x}} \rangle|}{\sqrt{S}\sqrt{\bm{x}^T\Sigma \bm{x}}} \rp^k \rp \\
        =& -\frac{1}{2}\log S - \frac{1}{2}\Phi^{-1}(\fpr_\lira(\bm{x}))^2 + \log\frac{|\langle \bm{x}, \bm{x} - \bm{m}_{\bm{x}} \rangle|}{\sqrt{\bm{x}^T\Sigma \bm{x}}\sqrt{2\pi}} \\
        &+ \log\lp 1 + \sum_{k=1}^\infty \frac{(-1)^k \He_k(\Phi^{-1}(\fpr_\lira(\bm{x})))}{(k+1)!} \lp \frac{|\langle \bm{x}, \bm{x} - \bm{m}_{\bm{x}} \rangle|}{\sqrt{S}\sqrt{\bm{x}^T\Sigma \bm{x}}} \rp^k \rp \\
        =& -\frac{1}{2}\log S - \frac{1}{2}\Phi^{-1}(\fpr_\lira(\bm{x}))^2 + \log\frac{|\langle \bm{x}, \bm{x} - \bm{m}_{\bm{x}} \rangle|}{\sqrt{\bm{x}^T\Sigma \bm{x}}\sqrt{2\pi}} + \log(1 + \xi(S)),
    \end{align}
    where we have $\xi(S) = O(1/\sqrt{S})$.
    For large enough $S$, ignoring the residual term, we approximate
    \begin{equation}\label{eq:approximation}
        \log(\tpr_\lira(\bm{x}) - \fpr_\lira(\bm{x})) \approx -\frac{1}{2}\log S - \frac{1}{2}\Phi^{-1}(\fpr_\lira(\bm{x}))^2 + \log\frac{|\langle \bm{x}, \bm{x} - \bm{m}_{\bm{x}} \rangle|}{\sqrt{\bm{x}^T \Sigma \bm{x}}\sqrt{2\pi}}.
      \end{equation}
\end{proof}

\subsection{Proof of \cref{cor:average-LiRA}}\label{sec:proof-average-lira}
\corAverageLiRA*
\begin{proof}
    By theorem \ref{theorem:LiRA-power-law} and the law of unconscious statistician, we have for large $S$
    \begin{align}
        \atpr_\lira - \afpr_\lira &= \int_\domain p(\bm{x}) (\tpr_\lira(\bm{x}) - \fpr_\lira(\bm{x}))\dr \bm{x} \\
        &\approx \int_\domain p(\bm{x})\frac{1}{\sqrt{2\pi}}e^{-\frac{1}{2}\Phi^{-1}(\afpr_\lira)^2}\frac{\langle \bm{x}, \bm{x} - \bm{m}_{\bm{x}} \rangle}{\sqrt{S}\sqrt{\bm{x}^T\Sigma \bm{x}}}\dr \bm{x} \\
        &= \frac{1}{\sqrt{2\pi}}e^{-\frac{1}{2}\Phi^{-1}(\afpr_\lira)^2}\frac{1}{\sqrt{S}} \int_\domain p(\bm{x})\frac{\langle \bm{x}, \bm{x} - \bm{m}_{\bm{x}} \rangle}{\sqrt{\bm{x}^T\Sigma \bm{x}}}\dr \bm{x} \\
        &= \frac{1}{\sqrt{S}}e^{-\frac{1}{2}\Phi^{-1}(\afpr_\lira)^2} \E_{(\bm{x}, y) \sim \sD}\left[\frac{\langle \bm{x}, \bm{x} - \bm{m}_{\bm{x}} \rangle}{\sqrt{2\pi}\sqrt{\bm{x}^T\Sigma \bm{x}}}\right],
    \end{align}
    where $p(\bm{x})$ is the density of $\sD$ at $(\bm{x}, y)$, and $\domain$ is the data domain.
    Note that here we fixed $\fpr_\lira(\bm{x}) = \afpr_\lira$ for all $\bm{x}$.
    Then we obtain
    \begin{equation}
        \log(\atpr_\lira - \afpr_\lira) \approx -\frac{1}{2}\log S - \frac{1}{2}\Phi^{-1}(\afpr_\lira)^2 + \log\lp \E_{(\bm{x}, y) \sim \sD}\left[\frac{\langle \bm{x}, \bm{x} - \bm{m}_{\bm{x}} \rangle}{\sqrt{2\pi}\sqrt{\bm{x}^T\Sigma \bm{x}}}\right] \rp.
    \end{equation}
\end{proof}

\section{Theoretical analysis of RMIA}\label{sec:rmia-theory}

Similar to LiRA, RMIA is based on shadow model training and computing the attack statistics based on a likelihood ratio. The main difference to LiRA is that RMIA does not compute the likelihood ratio based on aggregated IN/OUT statistics, but instead compares the target data point against random samples $(\bm{z}, y_{\bm{z}})$ from the target data distribution. After computing the likelihood ratios over multiple $(\bm{z}, y_{\bm{z}})$ values, the MIA score is estimated as a proportion of the ratios exceeding a preset bound. This approach makes RMIA a more effective attack when the number of shadow models is low.

\subsection{Formulating RMIA}

In the following, let us denote a target point by $(\bm{x}, y_{\bm{x}})$ with $y_{\bm{x}}$ being the label of $\bm{x}$.
Let us restate how RMIA \citep{Zarifzadeh2024RMIA} builds the MIA score. RMIA augments the likelihood-ratio with a sample from the target data distribution to calibrate how likely you would obtain the target model if $(\bm{x}, y_{\bm{x}})$ is replaced with another sample $(\bm{z}, y_{\bm{z}})$. Denoting the target model parameters with $\theta$, RMIA computes the \emph{pairwise} likelihood ratio
\begin{align}
    \LR(\bm{x}, \bm{z}) = \frac{p(\theta \mid \bm{x}, y_{\bm{x}})}{p(\theta \mid \bm{z}, y_{\bm{z}})},
\end{align}
and the corresponding MIA score is given as
\begin{align}
    \text{Score}_\text{RMIA}(\bm{x}) = \Pr_{(\bm{z}, y_{\bm{z}}) \sim \mathbb{D}}(\LR(\bm{x}, \bm{z}) > \gamma),
    \label{eq:rmia_score}
\end{align}
where $\mathbb{D}$ denotes the training data distribution.
\citet{Zarifzadeh2024RMIA} show two approaches to compute $\LR(\bm{x}, \bm{z})$: the direct approach and the effecient Bayesian approach.
In the following theoretical analysis, we focus on the direct approach that is an approximation of the efficient Bayesian approach, as \citet{Zarifzadeh2024RMIA} empirically demonstrates that these approaches exhibit similar performances.

Let $\hmu_{\bm{a},\bm{b}}$ and $\hsig_{\bm{a},\bm{b}}$ denote, respectively, the mean and standard deviation of $t_{\bm{b}} = \ell(\mathcal{M}(\bm{b}), y_{\bm{b}})$ estimated from shadow models, where $\bm{a}$ denotes which of $(\bm{x},y_{\bm{x}})$ and $(\bm{z},y_{\bm{z}})$ is in the training set.
By Equation 11 in \citep{Zarifzadeh2024RMIA}, the pairwise likelihood ratio is given as
\begin{equation}
    \LR(\bm{x}, \bm{z}) = \frac{p(\theta \mid \bm{x}, y_{\bm{x}})}{p(\theta \mid \bm{z}, y_{\bm{z}})} \approx \frac{\mathcal{N}(\tx; \hmu_{\bm{x},\bm{x}}, \hsig_{\bm{x},\bm{x}}^2) \mathcal{N}(\tz; \hmu_{\bm{x},\bm{z}}, \hsig_{\bm{x},\bm{z}}^2)}{\mathcal{N}(\tx; \hmu_{\bm{z},\bm{x}}, \hsig_{\bm{z},\bm{x}}^2) \mathcal{N}(\tz; \hmu_{\bm{z},\bm{z}}, \hsig_{\bm{z},\bm{z}}^2)},
\end{equation}
where $\hmu_{\bm{a},\bm{b}}$ and $\hsig_{\bm{a},\bm{b}}$ are, respectively, the mean and standard deviation of $t_b$ estimated from shadow models when the training set contains $\bm{a}$ but not $\bm{b}$.
Then RMIA exploits the probability of rejecting the pairwise likelihood ratio test over $(\bm{z}, y_{\bm{z}}) \sim \sD$:
\begin{equation}
    \mathrm{Score}_\rmia(\bm{x}) = \Pr_{(\bm{z}, y_{\bm{z}}) \sim \sD}\lp \mathrm{LR}(\bm{x}, \bm{z}) \geq \gamma \rp,
\end{equation}
where $\sD$ is the underlying data distribution.
Similar to LiRA, the classifier is built by checking if $\text{Score}_\text{RMIA}(\bm{x}) > \tau$. In the following, we will use the direct computation of likelihood-ratio as described in Equation 11 of \cite{Zarifzadeh2024RMIA} which approximates $\LR(\bm{x}, \bm{z})$ using normal distributions.
Thus, RMIA rejects $H_0$ if and only if
\begin{equation}
    \Pr_{(\bm{z}, y_{\bm{z}}) \sim \sD}\lp \mathrm{LR}(\bm{x}, \bm{z}) \geq \gamma \rp \geq \tau,
\end{equation}
identifying the membership of $\bm{x}$.
Hence the true positive rate of per-example RMIA is given as
\begin{multline}
    \tpr_\rmia(\bm{x}) \\= \Pr_{\dtarget \sim \sD^{|\data|}, \phi^M} \lp\Pr_{(\bm{z}, y_{\bm{z}}) \sim \sD}\lp \mathrm{LR}(\bm{x}, \bm{z}) \geq \gamma \rp \geq \tau \mid (\bm{x}, y_{\bm{x}}) \in \dtarget \wedge (\bm{x}, y_{\bm{x}}) \notin \dtarget\rp,
\end{multline}
where $\phi^M$ denotes the randomness in the shadow set sampling and shadow model training.

We define the average-case $\tpr$ for RMIA by taking the expectation over the data distribution:
\begin{equation}
    \atpr_\rmia = \E_{(\bm{x}, y_{\bm{x}}) \sim \sD}[\tpr_\rmia(\bm{x})].
\end{equation}

\subsection{Per-example RMIA}
Next we focus on the per-example RMIA performance.
As in the case of LiRA, we assume that $\tx$ and $\tz$ follow distributions of the location-scale family. 
We have
\begin{align}
    \tx &= \begin{cases}
        \mu_{\bm{x},\bm{x}} + \sigma_{\bm{x},\bm{x}}Z &\text{ if } (\bm{x}, y_{\bm{x}}) \in \dtarget \wedge (\bm{z}, y_{\bm{z}}) \notin \dtarget \\
        \mu_{\bm{z},\bm{x}} + \sigma_{\bm{z},\bm{x}}Z &\text{ if } (\bm{x}, y_{\bm{x}}) \notin \dtarget \wedge (\bm{z}, y_{\bm{z}}) \in \dtarget
    \end{cases} \\
    \tz &= \begin{cases}
        \mu_{\bm{x},\bm{z}} + \sigma_{\bm{x},\bm{z}}Z &\text{ if } (\bm{x}, y_{\bm{x}}) \in \dtarget \wedge (\bm{z}, y_{\bm{z}}) \notin \dtarget \\
        \mu_{\bm{z},\bm{z}} + \sigma_{\bm{z},\bm{z}}Z &\text{ if } (\bm{x}, y_{\bm{x}}) \notin \dtarget \wedge (\bm{z}, y_{\bm{z}}) \in \dtarget.
    \end{cases}
\end{align}
where $Z$ follows a distribution of location-scale family with the standard location and unit scale.
It is important to note that $\mu_{\bm{a},\bm{b}}$ and $\sigma_{\bm{a},\bm{b}}$ denote, respectively, a location and a scale, while previously defined $\hmu_{\bm{a},\bm{b}}$ and $\hsig_{\bm{a},\bm{b}}$ are, respectively, a mean and a standard deviation.
As for the analysis of LiRA, we assume that the target and shadow sets have a sufficient number of examples per class, and that $\sigma_{\bm{x}}=\sigma_{\bm{x},\bm{x}}=\sigma_{\bm{z},\bm{x}}$, $\sigma_{\bm{z}}=\sigma_{\bm{x},\bm{z}}=\sigma_{\bm{z},\bm{z}}$, $\hsig_{\bm{x}}=\hsig_{\bm{x},\bm{x}}=\hsig_{\bm{z},\bm{x}}$ and $\hsig_{\bm{z}}=\hsig_{\bm{x},\bm{z}}=\hsig_{z,z}$, where $\sigma_{\bm{x}}$ and $\sigma_{\bm{z}}$ are, respectively, the true scales of $\tx$ and $\tz$, and $\hsig_{\bm{x}}$ and $\hsig_{\bm{z}}$ are, respectively, standard deviations of $\tx$ and $\tz$ estimated from shadow models. Similarly to the case of LiRA, these assumptions are justified using the simplified model.
We have in the simplified model
\begin{align}
    \sigma_{\bm{x},\bm{x}}^2 &= \hsig_{\bm{x},\bm{x}}^2 = \Var(s_{y_{\bm{x}}}^{(\bm{x})}(\bm{x})) = \frac{1}{S}\lp1 - \frac{1}{S}\rp \bm{x}^T\Sigma \bm{x} \\
    \sigma_{\bm{z},\bm{x}}^2 &= \hsig_{\bm{z},\bm{x}}^2 = \Var(s_{y_{\bm{x}}}^{(\bm{z})}(\bm{x})) = \begin{cases}
        \frac{1}{S}\lp 1 - \frac{1}{S} \rp \bm{x}^T\Sigma \bm{x} \quad &\text{ if } y_{\bm{x}} = y_{\bm{z}} \\
        \frac{1}{S}\bm{x}^T\Sigma \bm{x} \quad &\text{ if } y_{\bm{x}} \not= y_{\bm{z}}
    \end{cases} \\
    \sigma_{\bm{x},\bm{z}}^2 &= \hsig_{\bm{x},\bm{z}}^2 = \Var(s_{y_{\bm{z}}}^{(\bm{x})}(\bm{z})) = \begin{cases}
        \frac{1}{S}\lp 1 - \frac{1}{S} \rp \bm{z}^T\Sigma \bm{z} \quad &\text{ if } y_{\bm{x}} = y_{\bm{z}} \\
        \frac{1}{S}\bm{z}^T\Sigma \bm{z} \quad &\text{ if } y_{\bm{x}} \not= y_{\bm{z}}
    \end{cases} \\
    \sigma_{\bm{z},\bm{z}}^2 &= \hsig_{\bm{z},\bm{z}}^2 = \Var(s_{y_{\bm{z}}}^{(\bm{z})}(\bm{z})) = \frac{1}{S}\lp 1 - \frac{1}{S} \rp \bm{z}^T\Sigma \bm{z}.
\end{align}
Therefore, the differences $\sigma_{\bm{x},\bm{x}}-\sigma_{\bm{z},\bm{x}}$, $\sigma_{\bm{x},\bm{z}}-\sigma_{\bm{z},\bm{z}}$, $\hsig_{\bm{x},\bm{x}}-\hsig_{\bm{z},\bm{x}}$ and $\hsig_{\bm{x},\bm{z}}-\hsig_{\bm{z},\bm{z}}$ are negligible for large enough $S$.

Now we derive the per-example RMIA vulnerability in terms of RMIA statistics computed from shadow models.

\begin{restatable}[Per-example RMIA vulnerability]{lemma}{lemmaPerExampleRMIA}\label{lemma:per-example-RMIA}
    Suppose that the true IN/OUT distributions of $\tx$ and $\tz$ are of location-scale family with locations $\mu_{\bm{x},\bm{x}}, \mu_{\bm{z},\bm{x}}, \mu_{\bm{x},\bm{z}}, \mu_{\bm{z},\bm{z}}$ and scales $\sigma_{\bm{x}}, \sigma_{\bm{z}}$. 
    Assume that RMIA models $\tx$ and $\tz$ by normal distributions with parameters computed from shadow models, and that an attacker has access to the underlying distribution $\sD$. Then with infinitely many shadow models, the RMIA vulnerability of a fixed target example $(\bm{x}, y_{\bm{x}})$ satisfies
    \begin{equation}
        \tpr_\rmia(\bm{x}) \leq 1 - F_{|Z|}\lp F_{|Z|}^{-1}(1 - \alpha) - \frac{\E_{(\bm{z}, y_{\bm{z}}) \sim \sD}[|q|]}{\E_{(\bm{z}, y_{\bm{z}}) \sim \sD}[|A|]} \rp
    \end{equation}
    for some constant $\alpha \geq \fpr_\rmia(\bm{x})$, where $F_{|Z|}$ is the cdf of $|Z|$ and
    \begin{align}
        q &= \frac{(\mu_{\bm{x},\bm{x}} - \mu_{\bm{z},\bm{x}})(\hmu_{\bm{x},\bm{x}} - \hmu_{\bm{z},\bm{x}})}{\hsig_{\bm{x}}^2} + \frac{(\mu_{\bm{x},\bm{z}} - \mu_{\bm{z},\bm{z}})(\hmu_{\bm{x},\bm{z}} - \hmu_{\bm{z},\bm{z}})}{\hsig_{\bm{z}}^2} \label{eq:rmia-px-param-q}\\
        A &= \frac{\sigma_{\bm{x}}}{\hsig_{\bm{x}}^2}(\hmu_{\bm{x},\bm{x}} - \hmu_{\bm{z},\bm{x}}) + \frac{\sigma_{\bm{z}}}{\hsig_{\bm{z}}^2}(\hmu_{\bm{x},\bm{z}} - \hmu_{\bm{z},\bm{z}}) \label{eq:rmia-px-param-A},
    \end{align}
    assuming that the inverse of $F_{|Z|}$ exists.
\end{restatable}
\begin{proof}
    We abuse notations by denoting probabilities and expectations over sampling $\dtarget \sim \sD^n$ and $(\bm{z}, y_{\bm{z}}) \sim \sD$ by subscripts $\dtarget$ and $\bm{z}$. We have in the limit of infinitely many shadow models
    \begin{align}
        \tpr_\rmia(\bm{x}) &= \Pr_{\dtarget} \lp\Pr_{\bm{z}}\lp \mathrm{LR}(\bm{x}, \bm{z}) \geq \gamma \rp \geq \tau \mid (\bm{x}, y_{\bm{x}}) \in \dtarget \wedge (\bm{z}, y_{\bm{z}}) \notin \dtarget\rp \\
        \fpr_\rmia(\bm{x}) &= \Pr_{\dtarget} \lp\Pr_{\bm{z}}\lp \mathrm{LR}(\bm{x}, \bm{z}) \geq \gamma \rp \geq \tau \mid (\bm{x}, y_{\bm{x}}) \notin \dtarget \wedge (\bm{z}, y_{\bm{z}}) \in \dtarget\rp,
    \end{align}
    where
    \begin{equation}
        \LR(\bm{x}, \bm{z}) = \frac{\mathcal{N}(\tx; \hmu_{\bm{x},\bm{x}}, \hsig_{\bm{x}}^2) \mathcal{N}(\tz; \hmu_{\bm{x},\bm{z}}, \hsig_{\bm{z}}^2)}{\mathcal{N}(\tx; \hmu_{\bm{z},\bm{x}}, \hsig_{\bm{x}}^2) \mathcal{N}(\tz; \hmu_{\bm{z},\bm{z}}, \hsig_{\bm{z}}^2)}.
    \end{equation}
    Note that the probabilities are over target dataset sampling in the limit of infinitely many shadow models. We have
    \begin{align}
        \lambda :=& \log(\mathrm{LR}(\bm{x},\bm{z})) \\
        =& \log\lp \frac{
        \frac{1}{\sqrt{2\pi}\hsig_{\bm{x}}}\exp\lp -\frac{1}{2}\lp \frac{\tx-\hmu_{\bm{x},\bm{x}}}{\hsig_{\bm{x}}} \rp^2 \rp \frac{1}{\sqrt{2\pi}\hsig_{\bm{z}}} \exp\lp -\frac{1}{2} \lp \frac{\tz-\hmu_{\bm{x},\bm{z}}}{\hsig_{\bm{z}}} \rp^2 \rp
        }{
        \frac{1}{\sqrt{2\pi}\hsig_{\bm{x}}}\exp\lp -\frac{1}{2}\lp \frac{\tx-\hmu_{\bm{z},\bm{x}}}{\hsig_{\bm{x}}} \rp^2 \rp \frac{1}{\sqrt{2\pi}\hsig_{\bm{z}}} \exp \lp -\frac{1}{2}\lp \frac{\tz-\hmu_{\bm{z},\bm{z}}}{\hsig_{\bm{z}}} \rp^2 \rp
        } \rp\\
        =& -\frac{1}{2}\lp \frac{\tx-\hmu_{\bm{x},\bm{x}}}{\hsig_{\bm{x}}} \rp^2 + \frac{1}{2}\lp \frac{\tx-\hmu_{\bm{z},\bm{x}}}{\hsig_{\bm{x}}} \rp^2 - \frac{1}{2}\lp \frac{\tz-\hmu_{\bm{x},\bm{z}}}{\hsig_{\bm{z}}} \rp^2 + \frac{1}{2}\lp \frac{\tz-\hmu_{\bm{z},\bm{z}}}{\hsig_{\bm{z}}} \rp^2 \\
        =& \frac{1}{2\hsig_{\bm{x}}^2}(2\tx\hmu_{\bm{x},\bm{x}} - \hmu_{\bm{x},\bm{x}}^2 - 2\tx\hmu_{\bm{z},\bm{x}} + \hmu_{\bm{z},\bm{x}}^2) + \frac{1}{2\hsig_{\bm{z}}^2}(2\tz\hmu_{\bm{x},\bm{z}} - \hmu_{\bm{x},\bm{z}}^2 - 2\tz\hmu_{\bm{z},\bm{z}} + \hmu_{\bm{z},\bm{z}}^2) \\
        =& \frac{\hmu_{\bm{x},\bm{x}} - \hmu_{\bm{z},\bm{x}}}{2\hsig_{\bm{x}}^2}(2\tx - \hmu_{\bm{x},\bm{x}} - \hmu_{\bm{z},\bm{x}}) + \frac{\hmu_{\bm{x},\bm{z}} - \hmu_{\bm{z},\bm{z}}}{2\hsig_{\bm{z}}^2}(2\tz - \hmu_{\bm{x},\bm{z}} - \hmu_{\bm{z},\bm{z}}). \label{eq:lambda}
    \end{align}
    When $(\bm{x}, y_{\bm{x}}) \in \dtarget$ and $(\bm{z}, y_{\bm{z}}) \notin \dtarget$, $\tx = \mu_{\bm{x},\bm{x}} + \sigma_{\bm{x}}Z$ and $\tz = \mu_{\bm{x},\bm{z}} + \sigma_{\bm{z}}Z$. Thus, $\lambda$ becomes
    \begin{align}
        \lambda_{\bm{x}} :=& \frac{\hmu_{\bm{x},\bm{x}} - \hmu_{\bm{z},\bm{x}}}{2\hsig_{\bm{x}}^2}(2\mu_{\bm{x},\bm{x}} + 2\sigma_{\bm{x}}Z - \hmu_{\bm{x},\bm{x}} - \hmu_{\bm{z},\bm{x}}) \\
        & + \frac{\hmu_{\bm{x},\bm{z}} - \hmu_{\bm{z},\bm{z}}}{2\hsig_{\bm{z}}^2}(2\mu_{\bm{x},\bm{z}} + 2\sigma_{\bm{z}}Z - \hmu_{\bm{x},\bm{z}} - \hmu_{\bm{z},\bm{z}}).
    \end{align}
    Similarly, when $(x, y_{\bm{x}}) \notin \dtarget$ and $(z, y_{\bm{z}}) \in \dtarget$, $\tx = \mu_{\bm{z},\bm{x}} + \sigma_{\bm{x}}Z$ and $\tz = \mu_{\bm{z},\bm{x}} + \sigma_{\bm{z}}Z$. Then $\lambda$ becomes
    \begin{align}
        \lambda_{\bm{z}} :=& \frac{\hmu_{\bm{x},\bm{x}} - \hmu_{\bm{z},\bm{x}}}{2\hsig_{\bm{x}}^2}(2\mu_{\bm{z},\bm{x}} + 2\sigma_{\bm{x}}Z - \hmu_{\bm{x},\bm{x}} - \hmu_{\bm{z},\bm{x}}) \\
        & + \frac{\hmu_{\bm{x},\bm{z}} - \hmu_{\bm{z},\bm{z}}}{2\hsig_{\bm{z}}^2}(2\mu_{\bm{z},\bm{z}} + 2\sigma_{\bm{z}}Z - \hmu_{\bm{x},\bm{z}} - \hmu_{\bm{z},\bm{z}}) \\
        =& \lp \underbrace{ \frac{\sigma_{\bm{x}}}{\hsig_{\bm{x}}^2}(\hmu_{\bm{x},\bm{x}} - \hmu_{\bm{z},\bm{x}}) + \frac{\sigma_{\bm{z}}}{\hsig_{\bm{z}}^2}(\hmu_{\bm{x},\bm{z}} - \hmu_{\bm{z},\bm{z}}) }_A \rp Z \\
        & + \underbrace{ \frac{\hmu_{\bm{x},\bm{x}} - \hmu_{\bm{z},\bm{x}}}{2\hsig_{\bm{x}}^2}(2\mu_{\bm{z},\bm{x}} - \hmu_{\bm{x},\bm{x}} - \hmu_{\bm{z},\bm{x}}) + \frac{\hmu_{\bm{x},\bm{z}} - \hmu_{\bm{z},\bm{z}}}{2\hsig_{\bm{z}}^2}(2\mu_{\bm{z},\bm{z}} - \hmu_{\bm{x},\bm{z}} - \hmu_{\bm{z},\bm{z}}) }_B \\
        =& AZ+B.
    \end{align}
    Notice that $A$ and $B$ are functions of $\bm{z}$ and independent of $Z$. Also note that taking probability over $Z$ corresponds to calculating probability over $\dtarget \sim \sD^n$. Thus, since we can set $\gamma > 1$, using Markov's inequality and the triangle inequality, we have in the limit of infinitely many shadow models
    \begin{align}
        \fpr_\rmia(x) = &\Pr_Z \lp\Pr_{\bm{z}}(\lambda_{\bm{z}} \geq \log\gamma) \geq \tau\rp \\
        &\leq \Pr_Z\lp\Pr_{\bm{z}}(|\lambda_{\bm{z}}| \geq \log\gamma) \geq \tau \rp \\
        &\leq \Pr_Z\lp\frac{\E_{\bm{z}}[|\lambda_{\bm{z}}|]}{\log\gamma}\geq \tau\rp \\ &=\Pr_Z\lp\E_{\bm{z}}[|\lambda_{\bm{z}}|]\geq \tau\log\gamma\rp \\
        &\leq \Pr_Z\lp\E_{\bm{z}}[|A|]|Z| + \E_{\bm{z}}[|B|]\geq \tau\log\gamma\rp \\
        &= \Pr_Z\lp|Z| \geq \frac{\tau\log\gamma - \E_{\bm{z}}[|B|]}{\E_{\bm{z}}[|A|]} \rp
    \end{align}
    Therefore, we can upper-bound $\fpr_\rmia(x)\leq\alpha$ by setting
    \begin{equation}
        \alpha = 1 - F_{|Z|}\lp\frac{\tau\log\gamma - \E_{\bm{z}}[|B|]}{\E_{\bm{z}}[|A|]}\rp.
    \end{equation}
    That is,
    \begin{equation}
        \frac{\tau\log\gamma - \E_{\bm{z}}[|B|]}{\E_{\bm{z}}[|A|]} = F_{|Z|}^{-1}(1-\alpha).
    \end{equation}
    Now let
    \begin{equation}
        q = \lambda_{\bm{x}} - \lambda_{\bm{z}} = \frac{(\mu_{\bm{x},\bm{x}} - \mu_{\bm{z},\bm{x}})(\hmu_{\bm{x},\bm{x}} - \hmu_{\bm{z},\bm{x}})}{\hsig_{\bm{x}}^2} + \frac{(\mu_{\bm{x},\bm{z}} - \mu_{\bm{z},\bm{z}})(\hmu_{\bm{x},\bm{z}} - \hmu_{\bm{z},\bm{z}})}{\hsig_{\bm{z}}^2}.
    \end{equation}
    Note that $q$ is also independent of $t$, thereby $\E_{\bm{z}}[|q|]$ being a constant. By the similar argument using Markov's inequality and the triangle inequality, we have
    \begin{align}
        \tpr_\rmia(\bm{x}) &= \Pr_Z \lp\Pr_{\bm{z}}(\lambda_{\bm{x}} \geq \log\gamma) \geq \tau\rp \\
        &\leq \Pr_Z\lp\Pr_{\bm{z}}(|\lambda_{\bm{x}}| \geq \log\gamma) \geq \tau \rp \\
        &\leq \Pr_Z\lp\frac{\E_{\bm{z}}[|\lambda_{\bm{x}}|]}{\log\gamma}\geq \tau\rp \\ &=\Pr_Z\lp\E_{\bm{z}}[|\lambda_{\bm{x}}|]\geq \tau\log\gamma\rp \\
        &\leq \Pr_Z\lp\E_{\bm{z}}[|A|]|Z| + \E_{\bm{z}}[|B|] + \E_{\bm{z}}[|q|] \geq \tau\log\gamma\rp \\
        &= \Pr_Z\lp|Z| \geq \frac{\tau\log\gamma - \E_{\bm{z}}[|B|]}{\E_{\bm{z}}[|A|]} - \frac{\E_{\bm{z}}[|q|]}{\E_{\bm{z}}[|A|]} \rp.
    \end{align}
    Hence we obtain
    \begin{equation}
        \tpr_\rmia(x) \leq 1 - F_{|Z|}^{-1}\lp F_{|Z|}^{-1}(1-\alpha) - \frac{\E_{\bm{z}}[|q|]}{\E_{\bm{z}}[|A|]} \rp.
    \end{equation}
\end{proof}

Note that, unlike per-example LiRA, we must assume that the attacker has access to the underlying distribution for the optimal RMIA as the \cref{eq:rmia-px-param-q,eq:rmia-px-param-A} depend on the parameters computed from shadow models.

\subsection{RMIA power-law}

Employing Lemma \ref{lemma:per-example-RMIA} and the simplified model, we obtain the following upper bound for RMIA performance.
\begin{restatable}[Per-example RMIA power-law]{theorem}{theoremRMIAPowerLaw}\label{theorem:RMIA-power-law}
    Fix a target example $(\bm{x},y_{\bm{x}})$. For the simplified model with large $S$ and infinitely many shadow models, the per-example RMIA vulnerability is given as
    \begin{equation}
        \tpr_\rmia(\bm{x}) \leq 1 - F_{|Z|}\lp F_{|Z|}^{-1}(1 - \alpha) - \frac{\Psi(\bm{x}, C)}{\sqrt{S}} \rp,
    \end{equation}
    where $\alpha \geq \fpr_\rmia(\bm{x})$, $F_{|Z|}$ is the cdf of the standard folded normal distribution, and
    \begin{equation}
        \Psi(\bm{x}, C) = \frac{
        \E_{\bm{z}}\left[\frac{\langle \bm{x}, \bm{x} - \bm{z}\rangle^2}{||\bm{x}||^2} + \frac{\langle \bm{z}, \bm{x} - \bm{z}\rangle^2}{||\bm{z}||^2} \mid y_{\bm{z}} = y_{\bm{x}} \right] + (C-1)\E_{\bm{z}}\left[\frac{\langle \bm{x}, \bm{x} - \bm{m}_{\bm{x}}\rangle^2}{||\bm{x}||^2} + \frac{\langle \bm{z}, \bm{z} - \bm{m}_{\bm{z}}\rangle^2}{||\bm{z}||^2} \mid y_{\bm{z}} \not= y_{\bm{x}}\right]
        }{
        \E_{\bm{z}}\left[\left\vert\frac{\langle \bm{x}, \bm{x} - \bm{z}\rangle}{||\bm{x}||} + \frac{\langle \bm{z}, \bm{x} - \bm{z}\rangle}{||\bm{z}||}\right\vert \mid y_{\bm{z}} = y_{\bm{x}}\right] + (C-1)\E_{\bm{z}}\left[\left\vert\frac{\langle \bm{x}, \bm{x} - \bm{m}_{\bm{x}}\rangle}{||\bm{x}||} + \frac{\langle \bm{z}, \bm{z} - \bm{m}_{\bm{z}}\rangle}{||\bm{z}||}\right\vert \mid y_{\bm{z}} \not= y_{\bm{x}}\right]
        }.
    \end{equation}
    In addition, we have
    \begin{equation}
        \log(\tpr_\rmia(\bm{x}) - \fpr_\rmia(\bm{x})) \approx -\frac{1}{2}\log S - \frac{1}{2}F_{|Z|}^{-1}(1 - \alpha)^2 + \log \frac{\Psi(\bm{x},C)}{\sqrt{\pi/2}}.
    \end{equation}
\end{restatable}

\begin{proof}
    To apply \cref{lemma:per-example-RMIA}, we will calculate $\E_{\bm{z}}[|q|]$ and $\E_{\bm{z}}[|A|]$. Let $\dtarget = \{(\bm{x}_{j,1}, j),\dots,(\bm{x}_{j,S}, j)\}_{j=1}^C$. Let $s_{y_{\bm{a}}}^{(\bm{b})}(\bm{a})$ denote the score of $(\bm{a}, y_{\bm{a}})$ when $\dtarget$ contains $(\bm{b}, y_{\bm{b}})$ but not the other example. Using similar argument as in the proof of \cref{theorem:LiRA-power-law}, we have
    \begin{align}
        s_{y_{\bm{x}}}^{(\bm{x})}(\bm{x}) &= \frac{1}{S}\langle \bm{x}, \sum_{i=1}^S \bm{x}_{y_{\bm{x}}, i} + \bm{x} - \bm{x}_{y_{\bm{x}}, S}\rangle \\
        s_{y_{\bm{x}}}^{(\bm{z})}(\bm{x}) &= \begin{cases}
            \frac{1}{S}\langle \bm{x}, \sum_{i=1}^S \bm{x}_{y_{\bm{x}}, i} + \bm{z} - \bm{x}_{y_{\bm{x}}, S}\rangle \quad &\text{ if } y_{\bm{x}} = y_{\bm{z}} \\
            \frac{1}{S}\langle \bm{x}, \sum_{i=1}^S \bm{x}_{y_{\bm{x}}, i}\rangle \quad &\text{ if } y_{\bm{x}} \not= y_{\bm{z}}
        \end{cases} \\
        s_{y_{\bm{z}}}^{(\bm{x})}(\bm{z}) &= \begin{cases}
            \frac{1}{S}\langle \bm{z}, \sum_{i=1}^{S} \bm{x}_{y_{\bm{x}}, i} + \bm{x} - \bm{x}_{y_{\bm{x}}, S}\rangle \quad &\text{ if } y_{\bm{x}} = y_{\bm{z}} \\
            \frac{1}{S}\langle \bm{z}, \sum_{i=1}^{S} \bm{x}_{y_{\bm{z}}, i}\rangle \quad &\text{ if } y_{\bm{x}} \not= y_{\bm{z}}
        \end{cases} \\
        s_{y_{\bm{z}}}^{(\bm{z})}(\bm{z}) &= \frac{1}{S}\langle \bm{z}, \sum_{i=1}^{S} \bm{x}_{y_{\bm{z}}, i} + \bm{z} - \bm{x}_{y_{\bm{z}}, S}\rangle.
    \end{align}
    Thus we obtain
    \begin{align}
        \mu_{\bm{x},\bm{x}} &= \langle \bm{x}, \bm{m}_{y_{\bm{x}}}\rangle + \frac{1}{S}\langle \bm{x}, \bm{x} - \bm{m}_{y_{\bm{x}}}\rangle \\
        \mu_{\bm{z},\bm{x}} &= \begin{cases}
            \langle \bm{x}, \bm{m}_{y_{\bm{x}}}\rangle + \frac{1}{S}\langle \bm{x}, \bm{z} - \bm{m}_{y_{\bm{x}}}\rangle \quad &\text{ if } y_{\bm{x}} = y_{\bm{z}} \\
            \langle \bm{x}, \bm{m}_{y_{\bm{x}}}\rangle \quad &\text{ if } y_{\bm{x}} \not= y_{\bm{z}}
        \end{cases} \\
        \mu_{\bm{x},\bm{z}} &= \begin{cases}
            \langle \bm{z}, \bm{m}_{y_{\bm{x}}}\rangle + \frac{1}{S}\langle \bm{z}, \bm{x} - \bm{m}_{y_{\bm{x}}}\rangle \quad &\text{ if } y_{\bm{x}} = y_{\bm{z}} \\
            \langle \bm{z}, \bm{m}_{y_{\bm{z}}}\rangle \quad &\text{ if } y_{\bm{x}} \not= y_{\bm{z}}
        \end{cases} \\
        \mu_{\bm{z},\bm{z}} &= \langle \bm{z}, \bm{m}_{y_{\bm{x}}}\rangle + \frac{1}{S}\langle \bm{z}, \bm{z} - \bm{m}_{y_{\bm{x}}}\rangle \quad \\
        \sigma_{\bm{x}} &= \frac{1}{\sqrt{S}}\sqrt{\bm{x}^T\Sigma \bm{x}} \\
        \sigma_{\bm{z}} &= \frac{1}{\sqrt{S}}\sqrt{\bm{z}^T\Sigma \bm{z}},
    \end{align}
    where $\bm{m}_{y_{\bm{z}}}$ is the true class mean of class $y_{\bm{z}}$.

    Now recall that
    \begin{align}
        q =& \frac{(\mu_{\bm{x},\bm{x}} - \mu_{\bm{z},\bm{x}})(\hmu_{\bm{x},\bm{x}} - \hmu_{\bm{z},\bm{x}})}{\hsig_{\bm{x}}^2} + \frac{(\mu_{\bm{x},\bm{z}} - \mu_{\bm{z},\bm{z}})(\hmu_{\bm{x},\bm{z}} - \hmu_{\bm{z},\bm{z}})}{\hsig_{\bm{z}}^2} \\
        A =& \frac{\sigma_{\bm{x}}}{\hsig_{\bm{x}}^2}(\hmu_{\bm{x},\bm{x}} - \hmu_{\bm{z},\bm{x}}) + \frac{\sigma_{\bm{z}}}{\hsig_{\bm{z}}^2}(\hmu_{\bm{x},\bm{z}} - \hmu_{\bm{z},\bm{z}}).
    \end{align}
    Since $\tx$ and $\tz$ follow normal distributions, the location and scale parameters of the true distributions correspond to the mean and standard deviations estimated from infinitely many shadow models, respectively. Thus, we have
    \begin{align}
        q =& \lp \frac{\mu_{\bm{x},\bm{x}} - \mu_{\bm{z},\bm{x}}}{\sigma_{\bm{x}}} \rp^2 + \lp \frac{\mu_{\bm{x},\bm{z}} - \mu_{\bm{z},\bm{z}}}{\sigma_{\bm{z}}} \rp^2 \\
        A =& \frac{\mu_{\bm{x},\bm{x}} - \mu_{\bm{z},\bm{x}}}{\sigma_{\bm{x}}} + \frac{\mu_{\bm{x},\bm{z}} - \mu_{\bm{z},\bm{z}}}{\sigma_{\bm{z}}}.
    \end{align}
    Using the law of total expectation, we have
    \begin{align}
        \E_{\bm{z}}[|q|] =& \Pr_{\bm{z}}(y_{\bm{z}}=y_{\bm{x}})\E_{\bm{z}}[|q| \mid y_{\bm{z}} = y_{\bm{x}}] + \sum_{j=1, j\not=y_{\bm{x}}}^C \Pr_{\bm{z}}(y_{\bm{z}}=j)\E_{\bm{z}}[|q| \mid y_{\bm{z}} = j] \\
        =& \frac{1}{C}\E_{\bm{z}}\left[\left.\lp\frac{\langle \bm{x}, \bm{x} - \bm{z}\rangle}{\sqrt{S}\sqrt{\bm{x}^T\Sigma \bm{x}}}\rp^2 + \lp\frac{\langle \bm{z}, \bm{x} - \bm{z}\rangle}{\sqrt{S}\sqrt{\bm{z}^T\Sigma \bm{z}}}\rp^2 \ \right\vert\ y_{\bm{z}} = y_{\bm{x}} \right] \\
        &+ \frac{C-1}{C}\E_{\bm{z}}\left[\left.\lp\frac{\langle \bm{x}, \bm{x} - \bm{m}_{y_{\bm{x}}}\rangle}{\sqrt{S}\sqrt{\bm{x}^T\Sigma \bm{x}}}\rp^2 + \lp\frac{\langle \bm{z}, \bm{z} - \bm{m}_{y_{\bm{z}}}\rangle}{\sqrt{S}\sqrt{\bm{z}^T\Sigma \bm{z}}}\rp^2 \ \right\vert\ y_{\bm{z}} \not= y_{\bm{x}}\right] \\
        =& \frac{1}{CS} \E_{\bm{z}}\left[\left.\frac{\langle \bm{x}, \bm{x} - \bm{z}\rangle^2}{\bm{x}^T\Sigma \bm{x}} + \frac{\langle \bm{z}, \bm{x} - \bm{z}\rangle^2}{\bm{z}^T\Sigma \bm{z}}\ \right\vert\ y_{\bm{z}} = y_{\bm{x}}\right] \\
        &+ \frac{C-1}{CS}\E_{\bm{z}}\left[\left.\frac{\langle \bm{x}, \bm{x} - \bm{m}_{y_{\bm{x}}}\rangle^2}{\bm{x}^T\Sigma \bm{x}} + \frac{\langle \bm{z}, \bm{z} - \bm{m}_{y_{\bm{z}}}\rangle^2}{\bm{z}^T\Sigma \bm{z}}\ \right\vert\ y_{\bm{z}} \not= y_{\bm{x}}\right],
    \end{align}
    and
    \begin{align}
        \E_{\bm{z}}[|A|] =& \Pr_{\bm{z}}(y_{\bm{z}}=y_{\bm{x}})\E_{\bm{z}}[|A| \mid y_{\bm{z}} = y_{\bm{x}}] + \sum_{j=1, j\not=y_{\bm{x}}}^C \Pr_{\bm{z}}(y_{\bm{z}}=j)\E_{\bm{z}}[|A| \mid y_{\bm{z}} = j] \\
        =& \frac{1}{C}\E_{\bm{z}}\left[\left.\left\vert\frac{\langle \bm{x}, \bm{x} - \bm{z}\rangle}{\sqrt{S}\sqrt{\bm{x}^T\Sigma \bm{x}}} + \frac{\langle \bm{z}, \bm{x} - \bm{z}\rangle}{\sqrt{S}\sqrt{\bm{z}^T\Sigma \bm{z}}}\right\vert \ \right\vert\ y_{\bm{z}} = y_{\bm{x}} \right] \\
        &+ \frac{C-1}{C}\E_{\bm{z}}\left[\left.\left\vert\frac{\langle \bm{x}, \bm{x} - \bm{m}_{\bm{x}}\rangle}{\sqrt{S}\sqrt{\bm{x}^T\Sigma \bm{x}}} + \frac{\langle \bm{z}, \bm{z} - \bm{m}_{\bm{z}}\rangle}{\sqrt{S}\sqrt{\bm{z}^T\Sigma \bm{z}}}\right\vert \ \right\vert\ y_{\bm{z}} \not= y_{\bm{x}}\right] \\
        =& \frac{1}{C\sqrt{S}} \E_{\bm{z}}\left[\left.\left\vert\frac{\langle \bm{x}, \bm{x} - \bm{z}\rangle}{\sqrt{\bm{x}^T\Sigma \bm{x}}} + \frac{\langle \bm{z}, \bm{x} - \bm{z}\rangle}{\sqrt{\bm{z}^T\Sigma \bm{z}}}\right\vert\ \right\vert\ y_{\bm{z}} = y_{\bm{x}}\right] \\
        &+ \frac{C-1}{C\sqrt{S}}\E_{\bm{z}}\left[\left.\left\vert\frac{\langle \bm{x}, \bm{x} - \bm{m}_{\bm{x}}\rangle}{\sqrt{\bm{x}^T\Sigma \bm{x}}} + \frac{\langle \bm{z}, \bm{z} - \bm{m}_{\bm{z}}\rangle}{\sqrt{\bm{z}^T\Sigma \bm{z}}}\right\vert\ \right\vert\ y_{\bm{z}} \not= y_{\bm{x}}\right].
    \end{align}
    Hence we obtain
    \begin{multline}
        \frac{\E_{\bm{z}}[|q|]}{\E_{\bm{z}}[|A|]} \\
        =\frac{1}{\sqrt{S}} \cdot \frac{
        \E_{\bm{z}}\left[\frac{\langle \bm{x}, \bm{x} - \bm{z}\rangle^2}{\bm{x}^T\Sigma \bm{x}} + \frac{\langle \bm{z}, \bm{x} - \bm{z}\rangle^2}{\bm{z}^T\Sigma \bm{z}} \mid y_{\bm{z}} = y_{\bm{x}} \right] + (C-1)\E_{\bm{z}}\left[\frac{\langle \bm{x}, \bm{x} - \bm{m}_{\bm{x}}\rangle^2}{\bm{x}^T\Sigma \bm{x}} + \frac{\langle \bm{z}, \bm{z} - \bm{m}_{\bm{z}}\rangle^2}{\bm{z}^T\Sigma \bm{z}} \mid y_{\bm{z}} \not= y_{\bm{x}}\right]
        }{
        \E_{\bm{z}}\left[\left\vert\frac{\langle \bm{x}, \bm{x} - \bm{z}\rangle}{\sqrt{\bm{x}^T\Sigma \bm{x}}} + \frac{\langle \bm{z}, \bm{x} - \bm{z}\rangle}{\sqrt{\bm{z}^T\Sigma \bm{z}}}\right\vert \mid y_{\bm{z}} = y_{\bm{x}}\right] + (C-1)\E_{\bm{z}}\left[\left\vert\frac{\langle \bm{x}, \bm{x} - \bm{m}_{\bm{x}}\rangle}{\sqrt{\bm{x}^T\Sigma \bm{x}}} + \frac{\langle \bm{z}, \bm{z} - \bm{m}_{\bm{z}}\rangle}{\sqrt{\bm{z}^T\Sigma \bm{z}}}\right\vert \mid y_{\bm{z}} \not= y_{\bm{x}}\right]
        }.
    \end{multline}
    Now \cref{lemma:per-example-RMIA} yields
    \begin{align}
        \tpr_\rmia(\bm{x}) &\leq 1 - F_{|Z|}\lp F_{|Z|}^{-1}(1 - \alpha) - \frac{\E_{\bm{z}}[|q|]}{\E_{\bm{z}}[|A|]} \rp \\
        &= 1 - F_{|Z|}\lp F_{|Z|}^{-1}(1 - \alpha) - \frac{\Psi(\bm{x}, C)}{\sqrt{S}} \rp,
    \end{align}
    where $F_{|Z|}$ is the cdf of the folded normal distribution and
    \begin{equation}
        \Psi(\bm{x}, C) = \frac{
        \E_{\bm{z}}\left[\frac{\langle \bm{x}, \bm{x} - \bm{z}\rangle^2}{\bm{x}^T\Sigma \bm{x}} + \frac{\langle \bm{z}, \bm{x} - \bm{z}\rangle^2}{\bm{z}^T\Sigma \bm{z}} \mid y_{\bm{z}} = y_{\bm{x}} \right] + (C-1)\E_{\bm{z}}\left[\frac{\langle \bm{x}, \bm{x} - \bm{m}_{\bm{x}}\rangle^2}{\bm{x}^T\Sigma \bm{x}} + \frac{\langle \bm{z}, \bm{z} - \bm{m}_{\bm{z}}\rangle^2}{\bm{z}^T\Sigma \bm{z}} \mid y_{\bm{z}} \not= y_{\bm{x}}\right]
        }{
        \E_{\bm{z}}\left[\left\vert\frac{\langle \bm{x}, \bm{x} - \bm{z}\rangle}{\sqrt{\bm{x}^T\Sigma \bm{x}}} + \frac{\langle \bm{z}, \bm{x} - \bm{z}\rangle}{\sqrt{\bm{z}^T\Sigma \bm{z}}}\right\vert \mid y_{\bm{z}} = y_{\bm{x}}\right] + (C-1)\E_{\bm{z}}\left[\left\vert\frac{\langle \bm{x}, \bm{x} - \bm{m}_{\bm{x}}\rangle}{\sqrt{\bm{x}^T\Sigma \bm{x}}} + \frac{\langle \bm{z}, \bm{z} - \bm{m}_{\bm{z}}\rangle}{\sqrt{\bm{z}^T\Sigma \bm{z}}}\right\vert \mid y_{\bm{z}} \not= y_{\bm{x}}\right]
        }.
    \end{equation}
    This completes the first half of the theorem.

    Now that from \cref{lemma:per-example-RMIA} we have
    \begin{align}
        \tpr_\rmia(\bm{x}) &= \Pr_Z\lp \Pr_{\bm{z}}(\lambda_{\bm{z}} + q \geq \log\gamma) \geq \tau \rp \\
        &\leq \Pr_Z\lp\frac{\E_{\bm{z}}[|\lambda_{\bm{z}}|] + \E_{\bm{z}}[|q|]}{\log\gamma} \geq \tau\rp \label{eq:tpr_inequality} \\
        &= \Pr_Z\lp |Z| \geq F_{|Z|}^{-1}(1 - \alpha) - \frac{\E_{\bm{z}}[|q|]}{\E_{\bm{z}}[|A|]} \rp \\
        \fpr_\rmia(\bm{x}) &= \Pr_Z\lp \Pr_{\bm{z}}(\lambda_{\bm{z}} \geq \log\gamma) \geq \tau\rp \\ &\leq\Pr_Z\lp\frac{\E_{\bm{z}}[|\lambda_{\bm{z}}|]}{\log\gamma} \geq \tau\rp \label{eq:fpr_inequality} \\
        &= \Pr_Z\lp |Z| \geq F_{|Z|}^{-1}(1 - \alpha) \rp. 
    \end{align}
    We claim that the bound for $\fpr_\rmia(\bm{x})$ (\cref{eq:fpr_inequality}) is as tight as that for $\tpr_\rmia(\bm{x})$ (\cref{eq:tpr_inequality}) for sufficiently large $S$. 
    Let us denote
    \begin{align}
        \kappa_\tpr(\gamma) &= \Pr_{\bm{z}}(\lambda_{\bm{z}} + q \geq \log\gamma) \\
        \kappa_\fpr(\gamma) &= \Pr_{\bm{z}}(\lambda_{\bm{z}} \geq \log\gamma).
    \end{align}
    Since $q=O(1/S)$, for large $S$ we approximate
    \begin{equation}
        \kappa_\tpr(\gamma) - \kappa_\fpr(\gamma) \approx p_{\lambda_{\bm{z}}}(\log\gamma)\frac{c_0}{S},
    \end{equation}
    for some constant $c_0$. Since $\lambda_{\bm{z}} = O(1/\sqrt{S})$, the scaled random variable $\hat{\lambda}_{\bm{z}} = \sqrt{S}\lambda_{\bm{z}}$ is almost independent of $S$. Then by the change of variables formula, we have
    \begin{equation}
        p_{\lambda_{\bm{z}}}(\log\gamma)\frac{c_0}{S} = p_{\hat{\lambda}_{\bm{z}}}(\sqrt{S}\log\gamma)\frac{c_0}{\sqrt{S}}.
    \end{equation}
    Without loss of generality we may set $\log\gamma = 1/\sqrt{S}$. Thus, we have
    \begin{equation}
        \kappa_\tpr(1/\sqrt{S}) - \kappa_\fpr(1/\sqrt{S}) \approx p_{\hat{\lambda}_{\bm{z}}}\lp\sqrt{S} \cdot \frac{1}{\sqrt{S}}\rp\frac{c_0}{\sqrt{S}} = p_{\hat{\lambda}_{\bm{z}}}(1)\frac{c_0}{\sqrt{S}}.
    \end{equation}
    This quantity scales as $O(1/\sqrt{S})$. Therefore,
    \begin{align}
        \tpr_\rmia(\bm{x}) - \fpr_\rmia(\bm{x}) &= \Pr_t(\kappa_\tpr(1/\sqrt{S}) \geq \tau) - \Pr_t(\kappa_\fpr(1/\sqrt{S}) \geq \tau) \\
        &\approx p_{\kappa_\fpr(1/\sqrt{S})}(\tau)(\kappa_\fpr(1/\sqrt{S}) - \kappa_\tpr(1/\sqrt{S})) \\
        &\approx p_{\kappa_\fpr(1/\sqrt{S})}(\tau)p_{\hat{\lambda}_{\bm{z}}}(1)\frac{c_0}{\sqrt{S}}.
    \end{align}
    Note that $\kappa_\fpr(1/\sqrt{S})$ and, consequently, $\tau$ are independent of $S$ for large enough $S$ since
    \begin{equation}
        \kappa_\fpr(1/\sqrt{S}) = \Pr_{\bm{z}}\lp\lambda_{\bm{z}} \geq \frac{1}{\sqrt{S}}\rp = \Pr_{\bm{z}}\lp\frac{\hat{\lambda}_{\bm{z}}}{\sqrt{S}} \geq \frac{1}{\sqrt{S}}\rp = \Pr_{\bm{z}}(\hat{\lambda}_{\bm{z}} \geq 1).
    \end{equation}
    Hence $\tpr_\rmia(\bm{x}) - \fpr_\rmia(\bm{x}) = O(1/\sqrt{S})$. On the other hand, from \cref{eq:tpr_inequality,eq:fpr_inequality} we have for sufficiently large $S$
    \begin{align}
        &\Pr_Z\lp\frac{\E_{\bm{z}}[|\lambda_{\bm{z}}|]}{\log\gamma} \geq \tau\rp - \Pr_Z\lp\frac{\E_{\bm{z}}[|\lambda_{\bm{z}}|]}{\log\gamma} + \frac{\E_{\bm{z}}[|q|]}{\log\gamma} \geq \tau\rp \\
        =& \Pr_Z\lp \sqrt{S}\E_{\bm{z}}[|\lambda_{\bm{z}}|] \geq \tau\rp - \Pr_Z\lp \sqrt{S}\E_{\bm{z}}[|\lambda_{\bm{z}}|] + \frac{c_1}{\sqrt{S}} \geq \tau\rp\\
        =& \Pr_Z\lp \E_{\bm{z}}[|\hat{\lambda}_{\bm{z}}|] \geq \tau\rp - \Pr_Z\lp \E_{\bm{z}}[|\hat{\lambda}_{\bm{z}}|] + \frac{c_1}{\sqrt{S}} \geq \tau\rp\\
        \approx&p_{\E_{\bm{z}}[|\hat{\lambda}_{\bm{z}}|]}(\tau)\frac{c_1}{\sqrt{S}}, \label{eq:rmia-bound-gap}
    \end{align}
    where $c_1$ is some constant. Since for large $S$, $\E_{\bm{z}}[|\hat{\lambda}_{\bm{z}}|]$ and $\tau$ are independent of $S$, \cref{eq:rmia-bound-gap} scales as $O(1/\sqrt{S})$.

    Now let
    \begin{align}
        \tpr_\lira(\bm{x}) &= \Pr_Z\lp\frac{\E_{\bm{z}}[|\lambda_{\bm{z}}|] + \E_{\bm{z}}[|q|]}{\log\gamma} \geq \tau\rp - v_\tpr \\
        \fpr_\lira(\bm{x}) &= \Pr_Z\lp\frac{\E_{\bm{z}}[|\lambda_{\bm{z}}|]}{\log\gamma} \geq \tau\rp - v_\fpr
    \end{align}
    for some $v_\tpr, v_\fpr \geq 0$ which evaluate the tightness of the bounds. Then we have
    \begin{align}
        &v_\tpr - v_\fpr \\
        &= \Pr_Z\lp\frac{\E_{\bm{z}}[|\lambda_{\bm{z}}|] + \E_{\bm{z}}[|q|]}{\log\gamma} \geq \tau\rp - \Pr_Z\lp\frac{\E_{\bm{z}}[|\lambda_{\bm{z}}|]}{\log\gamma} \geq \tau\rp - (\tpr_\lira(\bm{x}) - \fpr_\lira(\bm{x})) \\
        &= O(1/\sqrt{S}).
    \end{align}
    Hence we conclude that for sufficiently large $S$, the bound for $\fpr_\lira(\bm{x})$ is as tight as that for $\tpr_\lira(\bm{x})$.

    Therefore, noting that $\E_{\bm{z}}[|q|]/\E_{\bm{z}}[|A|]=O(1/\sqrt{S})$, for large $S$ we obtain
    \begin{align}
        &\tpr_\rmia(\bm{x}) - \fpr_\rmia(\bm{x}) \\
        &\approx \Pr_Z\lp |Z| \geq F_{|Z|}^{-1}(1 - \alpha) - \frac{\E_{\bm{z}}[|q|]}{\E_{\bm{z}}[|A|]} \rp - \Pr_Z\lp |Z| \geq F_{|Z|}^{-1}(1 - \alpha) \rp \\
        &\approx p_{|Z|}\lp F_{|Z|}^{-1}(1 - \alpha)\rp \frac{\E_{\bm{z}}[|q|]}{\E_{\bm{z}}[|A|]} \\
        &= p_{|Z|}\lp F_{|Z|}^{-1}(1 - \alpha)\rp \frac{\Psi(\bm{x}, C)}{\sqrt{S}}.
    \end{align}
    Since $|Z|$ follows the standard folded normal distribution,
    \begin{equation}
        p_{|Z|}\lp F_{|Z|}^{-1}(1 - \alpha)\rp = \frac{2}{\sqrt{2\pi}}\exp\lp -\frac{1}{2}F_{|Z|}^{-1}(1 - \alpha)^2\rp.
    \end{equation}
    It follows that
    \begin{equation}
        \log(\tpr_\rmia(\bm{x}) - \fpr_\rmia(\bm{x})) \approx -\frac{1}{2}\log S - \frac{1}{2}F_{|Z|}^{-1}(1 - \alpha)^2 + \log \frac{\Psi(\bm{x},C)}{\sqrt{\pi/2}}.
    \end{equation}
\end{proof}

As for the LiRA power-law, bounding $||\bm{x}-\bm{m}_{\bm{x}}||$ and $||\bm{z}-\bm{m}_{\bm{z}}||$ will provide a worst-case upper bound for which the power-law holds.

Finally, the per-example RMIA power-law is also extended to the average-case:
\begin{restatable}[Average-case RMIA power-law]{corollary}{corAverageRMIA}\label{cor:average-RMIA}
    For the simplified model with sufficiently large $S$ and infinitely many shadow models, we have
    \begin{equation}
        \log(\atpr_\rmia - \afpr_\rmia) \approx -\frac{1}{2}\log S - \frac{1}{2}F_{|Z|}^{-1}(1 - \alpha)^2 + \log\lp \mathop{\E}_{(\bm{x}, y_{\bm{x}}) \sim \sD}\left[\frac{\Psi(\bm{x}, C)}{\sqrt{\pi/2}} \right]\rp,
    \end{equation}
    where $\alpha \geq \afpr_\rmia$ and $F_{|Z|}$ is the cdf of the standard folded normal distribution.
\end{restatable}
\begin{proof}
    By \cref{theorem:RMIA-power-law} and the law of unconscious statistician, we have for large $S$
    \begin{align}
        \atpr_\rmia - \afpr_\rmia &= \int_\domain p(\bm{x}) (\tpr_\rmia(\bm{x}) - \fpr_\rmia(\bm{x}))\mathrm{d}\bm{x} \\
        &\approx \int_\domain p(\bm{x}) \frac{1}{\sqrt{\pi/2}}e^{-\frac{1}{2}F_{|Z|}^{-1}(1 - \alpha)^2}\frac{\Psi(\bm{x}, C)}{\sqrt{S}} \mathrm{d}\bm{x} \\
        &= \frac{1}{\sqrt{\pi/2}}e^{-\frac{1}{2}F_{|Z|}^{-1}(1 - \alpha)^2} \int_\domain p(\bm{x})\frac{\Psi(\bm{x}, C)}{\sqrt{S}} \mathrm{d}\bm{x}. \\
        &= \frac{1}{\sqrt{S}}e^{-\frac{1}{2}F_{|Z|}^{-1}(1 - \alpha)^2} \mathop{\E}_{(\bm{x}, y_{\bm{x}}) \sim \sD}\left[ \frac{\Psi(\bm{x}, C)}{\sqrt{\pi/2}} \right],
    \end{align}
    where $p(\bm{x})$ is the density of $\sD$ at $(\bm{x}, y_{\bm{x}})$, and $\domain$ is the data domain. Hence we obtain
    \begin{equation}
        \log(\atpr_\rmia - \afpr_\rmia) \approx -\frac{1}{2}\log S - \frac{1}{2}F_{|Z|}^{-1}(1 - \alpha)^2 + \log\lp \mathop{\E}_{(\bm{x}, y_{\bm{x}}) \sim \sD}\left[\frac{\Psi(\bm{x}, C)}{\sqrt{\pi/2}} \right]\rp.
    \end{equation}
\end{proof}

\clearpage

\section{Training details}\label{sec:training_details}

\subsection{Parameterization}\label{sec:parameterization}
We utilise pre-trained feature extractors BiT-M-R50x1 (R-50) \citep{kolesnikov2019big} with 23.5M parameters and Vision Transformer ViT-Base-16 (ViT-B) \citep{dosovitskiy2020image} with 85.8M parameters, both pretrained on the ImageNet-21K dataset~\citep{ILSVRC15}. We download the feature extractor checkpoints from the respective repositories.

Following \citet{tobaben2023Efficacy} that show the favorable trade-off of parameter-efficient fine-tuning between computational cost, utility and privacy even for small datasets, we only consider fine-tuning subsets of all feature extractor parameters. We consider the following configurations: 

\begin{itemize}[leftmargin=*,noitemsep,topsep=0pt]
\item  \textbf{Head:} We train a linear layer on top of the pre-trained feature extractor.
\item \textbf{FiLM:} In addition to the linear layer from Head, we fine-tune parameter-efficient FiLM~\citep{perez2018film} adapters scattered throughout the network. While a diverse set of adapters has been proposed, we utilise FiLM as it has been shown to be competitive in prior work~\citep{shysheya2022fit,tobaben2023Efficacy}.
\end{itemize}

\subsubsection{Licenses and access}~\label{sec:checkpoint_licenses}
The licenses and means to access the model checkpoints can be found below.

\begin{itemize}[leftmargin=*]
    \item BiT-M-R50x1 (R-50) \citep{kolesnikov2019big} is licensed with the Apache-2.0 license and can be obtained through the instructions on \url{https://github.com/google-research/big_transfer}.
    \item Vision Transformer ViT-Base-16 (ViT-B) \citep{dosovitskiy2020image} is licensed with the Apache-2.0 license and can be obtained through the instructions on \url{https://github.com/google-research/vision_transformer}.
\end{itemize}

\subsection{Hyperparameter tuning}\label{sec:hyperparameter_tuning}
Our hyperparameter tuning is heavily inspired by the comprehensive few-shot experiments by \citet{tobaben2023Efficacy}. We utilise their hyperparameter tuning protocol as it has been proven to yield SOTA results for (DP) few-shot models.
Given the input $\mathcal{\data}$ dataset we perform hyperparameter tuning by splitting the $\mathcal{\data}$ into 70$\%$ train and 30$\%$ validation. 
We then perform the specified iterations of hyperparameter tuning using the tree-structured Parzen estimator~\citep{bergstra_tpe_2011} strategy as implemented in Optuna~\citep{optuna_2019} to derive a set of hyperparameters that yield the highest accuracy on the validation split.
This set of hyperparameters is subsequently used to train all shadow models with the Adam optimizer~\citep{DBLP:journals/corr/KingmaB14}.
Details on the set of hyperparameters that are tuned and their ranges can be found in \cref{tab:hyperparam_ranges}.
\begin{table}[htbp]
	\caption{Hyperparameter ranges used for the Bayesian optimization with Optuna.}
	\label{tab:hyperparam_ranges}
	\centering
	\begin{small}
			\begin{tabular}{lcc}
				\toprule
				& \multicolumn{1}{l}{\textbf{lower bound}} & \multicolumn{1}{l}{\textbf{upper bound}} \\
				\midrule
    			batch size & \multicolumn{1}{r}{10} & \multicolumn{1}{r}{$|\mathcal{D}|$} \\
       			clipping norm & \multicolumn{1}{r}{0.2} & \multicolumn{1}{r}{10} \\
				epochs & \multicolumn{1}{r}{1} & \multicolumn{1}{r}{200} \\
				learning rate & \multicolumn{1}{r}{1e-7} & \multicolumn{1}{r}{1e-2} \\
				\bottomrule
			\end{tabular}
	\end{small}
\end{table}

\subsection{Datasets}\label{sec:datasets}
\cref{tab:used-datasets} shows the datasets used in the paper. We base our experiments on a subset of the the few-shot benchmark VTAB~\citep{zhai2019large} that achieves a classification accuracy $> 80\%$ and thus would considered to be used by a practitioner. Additionally, we add CIFAR10 which is not part of the original VTAB benchmark.

\begin{table}[!htbp]
\centering
\caption{Used datasets in the paper, their minimum and maximum shots $S$ and maximum number of classes $C$ and their test accuracy when fine-tuning a non-DP ViT-B Head. The test accuracy for EuroSAT and Resics45 is computed on the part of the training split that is not used for training the particular model due to both datasets missing an official test split. Note that LiRA requires $2S$ for training the shadow models and thus $S$ is smaller than when only performing fine-tuning.}
\begin{tabular}{lrrrrr}
\toprule
\textbf{dataset} & \textbf{(max.)} & \textbf{\boldmath min.} & \textbf{\boldmath max.} & \textbf{\boldmath test accuracy} & \textbf{\boldmath test accuracy} \\
& \textbf{\boldmath $C$} & \textbf{\boldmath $S$} & \textbf{\boldmath $S$} & \textbf{\boldmath (min. $S$)} & \textbf{\boldmath (max. $S$)} \\

\midrule
Patch Camelyon \citep{veeling2018rotation} & 2 & 256 & 65536 & 82.8\% & 85.6\% \\
CIFAR10 \citep{krizhevsky2009learning} & 10 & 8 & 2048 & 92.7\% & 97.7\% \\
EuroSAT \citep{helber2019eurosat}  & 10 & 8 & 512 & 80.2\% & 96.7\% \\
Pets \citep{parkhi2012cats} & 37 & 8 & 32 & 82.3\% & 90.7\%\\
Resisc45 \citep{cheng2017remote} & 45 & 32 & 256 & 83.5\% & 91.6\%\\
CIFAR100 \citep{krizhevsky2009learning}  & 100 & 16 & 128  & 82.2\% & 87.6\% \\
\bottomrule
\bottomrule
\end{tabular}
\label{tab:used-datasets}%
\end{table}%

\subsubsection{Licenses and access}~\label{sec:dataset_licenses}

The licenses and means to access the datasets can be found below. We downloaded all datasets from TensorFlow datasets \url{https://www.tensorflow.org/datasets} but Resics45 which required manual download.

\begin{itemize}[leftmargin=*]
    \item Patch Camelyon \citep{veeling2018rotation} is licensed with Creative Commons Zero v1.0 Universal (cc0-1.0) and we use version 2.0.0 of the dataset as specified on \url{https://www.tensorflow.org/datasets/catalog/patch_camelyon}.
    \item CIFAR10 \citep{krizhevsky2009learning} is licensed with an unknown license and we use version 3.0.2 of the dataset as specified on \url{https://www.tensorflow.org/datasets/catalog/cifar10}.
    \item EuroSAT \citep{helber2019eurosat} is licensed with MIT and we use version 2.0.0 of the dataset as specified on \url{https://www.tensorflow.org/datasets/catalog/eurosat}.
    \item Pets \citep{parkhi2012cats} is licensed with CC BY-SA 4.0 Deed and we use version 3.2.0 of the dataset as specified on \url{https://www.tensorflow.org/datasets/catalog/oxford_iiit_pet}.
    \item Resisc45 \citep{cheng2017remote} is licensed with an unknown license and we use version 3.0.0 of the dataset as specified on \url{https://www.tensorflow.org/datasets/catalog/resisc45}.
    \item CIFAR100 \citep{krizhevsky2009learning} is licensed with an unknown license and we use version 3.0.2 of the dataset as specified on \url{https://www.tensorflow.org/datasets/catalog/cifar100}.
\end{itemize}

\subsection{Compute resources}~\label{sec:compute_resources}

All experiments but the R-50 (FiLM) experiments are run on CPU with 8 cores and 16 GB of host memory. The training time depends on the model (ViT is cheaper than R-50), number of shots $\shots$ and the number of classes $\classes$ but ranges for the training of one model from some minutes to an hour. This assumes that the images are passed once through the pre-trained backbone and then cached as feature vectors. The provided code implements this optimization.

The R-50 (FiLM) experiments are significantly more expensive and utilise a NVIDIA V100 with 40 GB VRAM, 10 CPU cores and 64 GB of host memory. The training of 257 shadow models then does not exceed 24h for the settings that we consider.

We estimate that in total we spend around 7 days of V100 and some dozens of weeks of CPU core time but more exact measurements are hard to make.

\clearpage

\section{Additional results}
In this section, we provide tabular results for our experiments and additional figures that did not fit into the main paper.

\subsection{Additional results for \cref{sec:predicting_dataset}}\label{sec:additional_sec5}
This Section contains additional results for \cref{sec:predicting_dataset}.

\subsubsection{Vulnerability as a function of shots}
This section displays additional results to \cref{fig:fig1} for $\fpr$ $\in \{0.1, 0.01, 0.001\}$ for ViT-B and R-50 in in \cref{fig:mia_function_shots_additional,tab:vit-function-shots,tab:r-50-function-shots}.

\begin{figure*}[h]
    \centering
    \begin{subfigure}[b]{0.49\textwidth}
        \includegraphics[width=\textwidth]{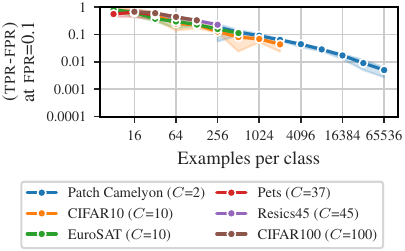}
        \caption{ViT-B Head $\tpr-\fpr$ at $\fpr=0.1$}
    \end{subfigure}
    \begin{subfigure}[b]{0.49\textwidth}
        \includegraphics[width=\textwidth]{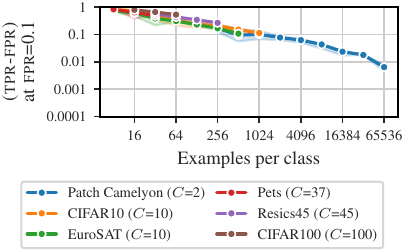}
        \caption{R-50 Head $\tpr-\fpr$ at $\fpr=0.1$}
    \end{subfigure}

    \vspace{0.5cm}
    
    \begin{subfigure}[b]{0.49\textwidth}
        \includegraphics[width=\textwidth]{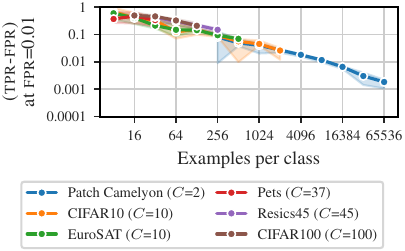}
        \caption{ViT-B Head $\tpr-\fpr$ at $\fpr=0.01$}
    \end{subfigure}
    \begin{subfigure}[b]{0.49\textwidth}
        \includegraphics[width=\textwidth]{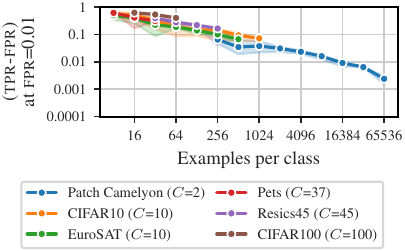}
        \caption{R-50 Head $\tpr-\fpr$ at $\fpr=0.01$}
    \end{subfigure}

    \vspace{0.5cm}

    \begin{subfigure}[b]{0.49\textwidth}
        \includegraphics[width=\textwidth]{figures/section4/mia_as_function_of_shots_0.001_ViT-B.pdf}
        \caption{ViT-B Head $\tpr-\fpr$ at $\fpr=0.001$}
    \end{subfigure}    
    \begin{subfigure}[b]{0.49\textwidth}
        \includegraphics[width=\textwidth]{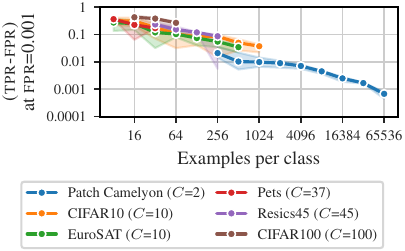}
        \caption{R-50 Head $\tpr-\fpr$ at $\fpr=0.001$}
    \end{subfigure}
    \caption{MIA vulnerability as a function of shots (examples per class) when attacking a pre-trained ViT-B and R-50 Head trained without DP on different downstream datasets. The errorbars display the minimum and maximum Clopper-Pearson CIs over six seeds and the solid line the median.}
    \label{fig:mia_function_shots_additional}
\end{figure*}

\begin{table}[ht]
  \centering
  \caption{Median MIA vulnerability over six seeds as a function of $S$ (shots) when attacking a Head trained without DP on-top of a ViT-B. The ViT-B is pre-trained on ImageNet-21k.}
\begin{adjustbox}{max width=\textwidth}
\begin{tabular}{lrrrrr}
\toprule
\textbf{dataset} & \textbf{\boldmath classes ($C$)} & \textbf{\boldmath shots ($S$)} & \textbf{tpr@fpr=0.1} & \textbf{tpr@fpr=0.01} & \textbf{tpr@fpr=0.001} \\
\midrule
\multirow[t]{9}{*}{Patch Camelyon \citep{veeling2018rotation}} & \multirow[t]{9}{*}{2} & 256 & 0.266 & 0.086 & 0.032 \\
 &  & 512 & 0.223 & 0.059 & 0.018 \\
 &  & 1024 & 0.191 & 0.050 & 0.015 \\
 &  & 2048 & 0.164 & 0.037 & 0.009 \\
 &  & 4096 & 0.144 & 0.028 & 0.007 \\
 &  & 8192 & 0.128 & 0.021 & 0.005 \\
 &  & 16384 & 0.118 & 0.017 & 0.003 \\
 &  & 32768 & 0.109 & 0.014 & 0.002 \\
 &  & 65536 & 0.105 & 0.012 & 0.002 \\
\cline{1-6} \cline{2-6}
\multirow[t]{9}{*}{CIFAR10 \citep{krizhevsky2009learning}} & \multirow[t]{9}{*}{10} & 8 & 0.910 & 0.660 & 0.460 \\
 &  & 16 & 0.717 & 0.367 & 0.201 \\
 &  & 32 & 0.619 & 0.306 & 0.137 \\
 &  & 64 & 0.345 & 0.132 & 0.067 \\
 &  & 128 & 0.322 & 0.151 & 0.082 \\
 &  & 256 & 0.227 & 0.096 & 0.054 \\
 &  & 512 & 0.190 & 0.068 & 0.032 \\
 &  & 1024 & 0.168 & 0.056 & 0.025 \\
 &  & 2048 & 0.148 & 0.039 & 0.013 \\
\cline{1-6} \cline{2-6}
\multirow[t]{7}{*}{EuroSAT \citep{helber2019eurosat}} & \multirow[t]{7}{*}{10} & 8 & 0.921 & 0.609 & 0.408 \\
 &  & 16 & 0.738 & 0.420 & 0.234 \\
 &  & 32 & 0.475 & 0.222 & 0.113 \\
 &  & 64 & 0.400 & 0.159 & 0.074 \\
 &  & 128 & 0.331 & 0.155 & 0.084 \\
 &  & 256 & 0.259 & 0.104 & 0.049 \\
 &  & 512 & 0.213 & 0.080 & 0.037 \\
\cline{1-6} \cline{2-6}
\multirow[t]{3}{*}{Pets \citep{parkhi2012cats}} & \multirow[t]{3}{*}{37} & 8 & 0.648 & 0.343 & 0.160 \\
 &  & 16 & 0.745 & 0.439 & 0.259 \\
 &  & 32 & 0.599 & 0.311 & 0.150 \\
\cline{1-6} \cline{2-6}
\multirow[t]{4}{*}{Resics45 \citep{cheng2017remote}} & \multirow[t]{4}{*}{45} & 32 & 0.672 & 0.425 & 0.267 \\
 &  & 64 & 0.531 & 0.295 & 0.168 \\
 &  & 128 & 0.419 & 0.212 & 0.115 \\
 &  & 256 & 0.323 & 0.146 & 0.072 \\
\cline{1-6} \cline{2-6}
\multirow[t]{4}{*}{CIFAR100 \citep{krizhevsky2009learning}} & \multirow[t]{4}{*}{100} & 16 & 0.814 & 0.508 & 0.324 \\
 &  & 32 & 0.683 & 0.445 & 0.290 \\
 &  & 64 & 0.538 & 0.302 & 0.193 \\
 &  & 128 & 0.433 & 0.208 & 0.114 \\
\cline{1-6} \cline{2-6}
\bottomrule
\end{tabular}
  \label{tab:vit-function-shots}%
\end{adjustbox}
\end{table}%

\begin{table}[ht]
  \centering
  \caption{Median MIA vulnerability over six seeds as a function of $S$ (shots) when attacking a Head trained without DP on-top of a R-50. The R-50 is pre-trained on ImageNet-21k.}
\begin{adjustbox}{max width=\textwidth}
\begin{tabular}{lrrrrr}
\toprule
\textbf{dataset} & \textbf{\boldmath classes ($C$)} & \textbf{\boldmath shots ($S$)} & \textbf{tpr@fpr=0.1} & \textbf{tpr@fpr=0.01} & \textbf{tpr@fpr=0.001} \\
\midrule
\multirow[t]{9}{*}{Patch Camelyon \citep{veeling2018rotation}} & \multirow[t]{9}{*}{2} & 256 & 0.272 & 0.076 & 0.022 \\
 &  & 512 & 0.195 & 0.045 & 0.011 \\
 &  & 1024 & 0.201 & 0.048 & 0.011 \\
 &  & 2048 & 0.178 & 0.041 & 0.010 \\
 &  & 4096 & 0.163 & 0.033 & 0.008 \\
 &  & 8192 & 0.143 & 0.026 & 0.006 \\
 &  & 16384 & 0.124 & 0.019 & 0.004 \\
 &  & 32768 & 0.118 & 0.016 & 0.003 \\
 &  & 65536 & 0.106 & 0.012 & 0.002 \\
\cline{1-6} \cline{2-6}
\multirow[t]{8}{*}{CIFAR10 \citep{krizhevsky2009learning}} & \multirow[t]{8}{*}{10} & 8 & 0.911 & 0.574 & 0.324 \\
 &  & 16 & 0.844 & 0.526 & 0.312 \\
 &  & 32 & 0.617 & 0.334 & 0.183 \\
 &  & 64 & 0.444 & 0.208 & 0.106 \\
 &  & 128 & 0.334 & 0.159 & 0.084 \\
 &  & 256 & 0.313 & 0.154 & 0.086 \\
 &  & 512 & 0.251 & 0.103 & 0.051 \\
 &  & 1024 & 0.214 & 0.082 & 0.038 \\
\cline{1-6} \cline{2-6}
\multirow[t]{7}{*}{EuroSAT \citep{helber2019eurosat}} & \multirow[t]{7}{*}{10} & 8 & 0.846 & 0.517 & 0.275 \\
 &  & 16 & 0.699 & 0.408 & 0.250 \\
 &  & 32 & 0.490 & 0.236 & 0.121 \\
 &  & 64 & 0.410 & 0.198 & 0.105 \\
 &  & 128 & 0.332 & 0.151 & 0.075 \\
 &  & 256 & 0.269 & 0.111 & 0.056 \\
 &  & 512 & 0.208 & 0.077 & 0.036 \\
\cline{1-6} \cline{2-6}
\multirow[t]{3}{*}{Pets \citep{parkhi2012cats}} & \multirow[t]{3}{*}{37} & 8 & 0.937 & 0.631 & 0.366 \\
 &  & 16 & 0.745 & 0.427 & 0.227 \\
 &  & 32 & 0.588 & 0.321 & 0.173 \\
\cline{1-6} \cline{2-6}
\multirow[t]{4}{*}{Resics45 \citep{cheng2017remote}} & \multirow[t]{4}{*}{45} & 32 & 0.671 & 0.405 & 0.235 \\
 &  & 64 & 0.534 & 0.289 & 0.155 \\
 &  & 128 & 0.445 & 0.231 & 0.121 \\
 &  & 256 & 0.367 & 0.177 & 0.088 \\
\cline{1-6} \cline{2-6}
\multirow[t]{3}{*}{CIFAR100 \citep{krizhevsky2009learning}} & \multirow[t]{3}{*}{100} & 16 & 0.897 & 0.638 & 0.429 \\
 &  & 32 & 0.763 & 0.549 & 0.384 \\
 &  & 64 & 0.634 & 0.414 & 0.269 \\
\cline{1-6} \cline{2-6}
\bottomrule
\end{tabular}
  \label{tab:r-50-function-shots}%
\end{adjustbox}
\end{table}%

\clearpage

\subsubsection{Vulnerability as a function of the number of classes}

This section displays additional results to \cref{fig:mia_function_classes} for $\fpr$ $\in \{0.1, 0.01, 0.001\}$ for ViT-B and R-50 in in \cref{fig:mia_function_classes_additional,tab:vit-b-function-classes,tab:r-50-function-classes}.

\begin{figure*}[h]
    \centering
    \begin{subfigure}[b]{0.49\textwidth}
        \includegraphics[width=\textwidth]{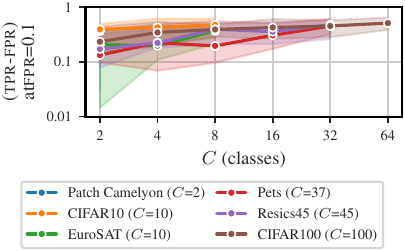}
        \caption{ViT-B Head $\tpr-\fpr$ at $\fpr=0.1$}
    \end{subfigure}
    \begin{subfigure}[b]{0.49\textwidth}
        \includegraphics[width=\textwidth]{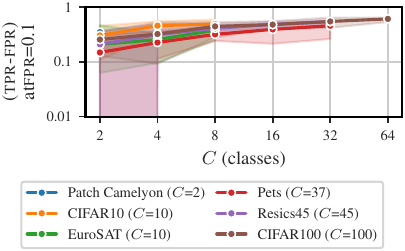}
        \caption{R-50 Head $\tpr-\fpr$ at $\fpr=0.1$}
    \end{subfigure}

    \vspace{0.5cm}
    
    \begin{subfigure}[b]{0.49\textwidth}
        \includegraphics[width=\textwidth]{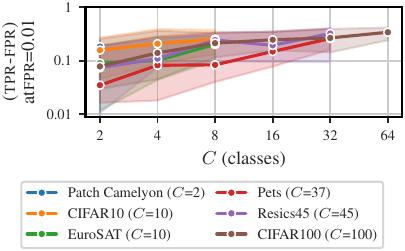}
        \caption{ViT-B Head $\tpr-\fpr$ at $\fpr=0.01$}
    \end{subfigure}
    \begin{subfigure}[b]{0.49\textwidth}
        \includegraphics[width=\textwidth]{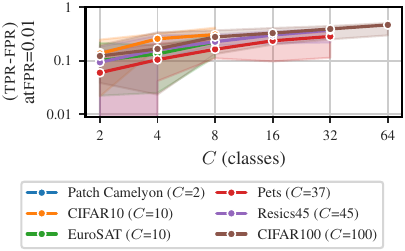}
        \caption{R-50 Head $\tpr-\fpr$ at $\fpr=0.01$}
    \end{subfigure}

    \vspace{0.5cm}

    \begin{subfigure}[b]{0.49\textwidth}
        \includegraphics[width=\textwidth]{figures/section4/mia_as_function_of_classes_0.001_ViT-B.pdf}
        \caption{ViT-B Head $\tpr-\fpr$ at $\fpr=0.001$}
    \end{subfigure}    
    \begin{subfigure}[b]{0.49\textwidth}
        \includegraphics[width=\textwidth]{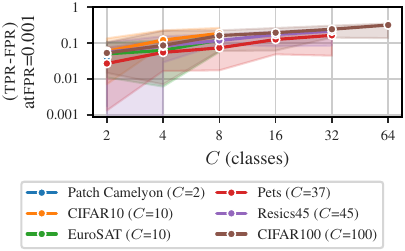}
        \caption{R-50 Head $\tpr-\fpr$ at $\fpr=0.001$}
    \end{subfigure}
    \caption{MIA vulnerability as a function of $\classes$ (classes) when attacking a ViT-B and R-50 Head fine-tuned without DP on different datasets where the classes are randomly sub-sampled and $S=32$. The solid line displays the median and the errorbars the min and max clopper-pearson CIs over 12 seeds.}
    \label{fig:mia_function_classes_additional}
\end{figure*}

\begin{table}[ht]
  \centering
  \caption{Median MIA vulnerability over 12 seeds as a function of $C$ (classes) when attacking a Head trained without DP on-top of a ViT-B. The Vit-B is pre-trained on ImageNet-21k.}
\begin{adjustbox}{max width=\textwidth}
\begin{tabular}{lrrrrr}
\toprule
\textbf{dataset} & \textbf{\boldmath shots ($S$)} & \textbf{\boldmath classes ($C$)} & \textbf{tpr@fpr=0.1} & \textbf{tpr@fpr=0.01} & \textbf{tpr@fpr=0.001} \\
\midrule
Patch Camelyon \citep{veeling2018rotation} & 32 & 2 & 0.467 & 0.192 & 0.080 \\
\cline{1-6} \cline{2-6}
\multirow[t]{3}{*}{CIFAR10 \citep{krizhevsky2009learning}} & \multirow[t]{3}{*}{32} & 2 & 0.494 & 0.167 & 0.071 \\
 &  & 4 & 0.527 & 0.217 & 0.115 \\
 &  & 8 & 0.574 & 0.262 & 0.123 \\
\cline{1-6} \cline{2-6}
\multirow[t]{3}{*}{EuroSAT \citep{helber2019eurosat}} & \multirow[t]{3}{*}{32} & 2 & 0.306 & 0.100 & 0.039 \\
 &  & 4 & 0.298 & 0.111 & 0.047 \\
 &  & 8 & 0.468 & 0.211 & 0.103 \\
\cline{1-6} \cline{2-6}
\multirow[t]{5}{*}{Pets \citep{parkhi2012cats}} & \multirow[t]{5}{*}{32} & 2 & 0.232 & 0.045 & 0.007 \\
 &  & 4 & 0.324 & 0.092 & 0.035 \\
 &  & 8 & 0.296 & 0.094 & 0.035 \\
 &  & 16 & 0.406 & 0.158 & 0.069 \\
 &  & 32 & 0.553 & 0.269 & 0.136 \\
\cline{1-6} \cline{2-6}
\multirow[t]{5}{*}{Resics45 \citep{cheng2017remote}} & \multirow[t]{5}{*}{32} & 2 & 0.272 & 0.084 & 0.043 \\
 &  & 4 & 0.322 & 0.119 & 0.056 \\
 &  & 8 & 0.496 & 0.253 & 0.148 \\
 &  & 16 & 0.456 & 0.204 & 0.108 \\
 &  & 32 & 0.580 & 0.332 & 0.195 \\
\cline{1-6} \cline{2-6}
\multirow[t]{6}{*}{CIFAR100 \citep{krizhevsky2009learning}} & \multirow[t]{6}{*}{32} & 2 & 0.334 & 0.088 & 0.035 \\
 &  & 4 & 0.445 & 0.150 & 0.061 \\
 &  & 8 & 0.491 & 0.223 & 0.121 \\
 &  & 16 & 0.525 & 0.256 & 0.118 \\
 &  & 32 & 0.553 & 0.276 & 0.153 \\
 &  & 64 & 0.612 & 0.350 & 0.211 \\
\cline{1-6} \cline{2-6}
\bottomrule
\end{tabular}
\label{tab:vit-b-function-classes}%
\end{adjustbox}
\end{table}%

\begin{table}[ht]
  \centering
  \caption{Median MIA vulnerability over 12 seeds as a function of $C$ (classes) when attacking a Head trained without DP on-top of a R-50. The R-50 is pre-trained on ImageNet-21k.}
\begin{adjustbox}{max width=\textwidth}
\begin{tabular}{lrrrrr}
\toprule
\textbf{dataset} & \textbf{\boldmath shots ($S$)} & \textbf{\boldmath classes ($C$)} & \textbf{tpr@fpr=0.1} & \textbf{tpr@fpr=0.01} & \textbf{tpr@fpr=0.001} \\
\midrule
Patch Camelyon \citep{veeling2018rotation} & 32 & 2 & 0.452 & 0.151 & 0.041 \\
\cline{1-6} \cline{2-6}
\multirow[t]{3}{*}{CIFAR10 \citep{krizhevsky2009learning}} & \multirow[t]{3}{*}{32} & 2 & 0.404 & 0.146 & 0.060 \\
 &  & 4 & 0.560 & 0.266 & 0.123 \\
 &  & 8 & 0.591 & 0.318 & 0.187 \\
\cline{1-6} \cline{2-6}
\multirow[t]{3}{*}{EuroSAT \citep{helber2019eurosat}} & \multirow[t]{3}{*}{32} & 2 & 0.309 & 0.111 & 0.050 \\
 &  & 4 & 0.356 & 0.144 & 0.064 \\
 &  & 8 & 0.480 & 0.233 & 0.123 \\
\cline{1-6} \cline{2-6}
\multirow[t]{5}{*}{Pets \citep{parkhi2012cats}} & \multirow[t]{5}{*}{32} & 2 & 0.249 & 0.068 & 0.029 \\
 &  & 4 & 0.326 & 0.115 & 0.056 \\
 &  & 8 & 0.419 & 0.173 & 0.075 \\
 &  & 16 & 0.493 & 0.245 & 0.127 \\
 &  & 32 & 0.559 & 0.294 & 0.166 \\
\cline{1-6} \cline{2-6}
\multirow[t]{5}{*}{Resics45 \citep{cheng2017remote}} & \multirow[t]{5}{*}{32} & 2 & 0.310 & 0.103 & 0.059 \\
 &  & 4 & 0.415 & 0.170 & 0.083 \\
 &  & 8 & 0.510 & 0.236 & 0.119 \\
 &  & 16 & 0.585 & 0.311 & 0.174 \\
 &  & 32 & 0.644 & 0.382 & 0.218 \\
\cline{1-6} \cline{2-6}
\multirow[t]{6}{*}{CIFAR100 \citep{krizhevsky2009learning}} & \multirow[t]{6}{*}{32} & 2 & 0.356 & 0.132 & 0.054 \\
 &  & 4 & 0.423 & 0.176 & 0.087 \\
 &  & 8 & 0.545 & 0.288 & 0.163 \\
 &  & 16 & 0.580 & 0.338 & 0.196 \\
 &  & 32 & 0.648 & 0.402 & 0.244 \\
 &  & 64 & 0.711 & 0.476 & 0.320 \\
\cline{1-6} \cline{2-6}
\bottomrule
\end{tabular}
\label{tab:r-50-function-classes}%
\end{adjustbox}
\end{table}%

\clearpage

\subsubsection{Data for FiLM and from scratch training}

\begin{table}[ht]
  \centering
  \caption{MIA vulnerability data used in \cref{fig:eval_r50filmcarlini}. Note that the data from \citet{Carlini2022LiRA} is only partially tabular, thus we estimated the $\tpr$ at $\fpr$ from the plots in the Appendix of their paper.}\label{tab:comparision_data}
\begin{adjustbox}{max width=\textwidth}
\begin{tabular}{llllllll}
\toprule
\textbf{model} & 
\textbf{dataset} & \textbf{\boldmath classes}
 &
\textbf{\boldmath shots}  &
\textbf{source}
&
\textbf{tpr@} & \textbf{tpr@} & \textbf{tpr@}  \\
& &  \textbf{\boldmath ($C$)}
 &
\textbf{\boldmath ($S$)} & & \textbf{fpr=0.1} & \textbf{fpr=0.01}& \textbf{fpr=0.001}\\
\midrule
\multirow[t]{9}{*}{R-50 FiLM} & CIFAR10 & 10 & 50 & This work & 0.482 & 0.275 & 0.165 \\
 & \citep{krizhevsky2009learning} & & &  &  & &  \\
\cline{2-8} \cline{3-8} \cline{4-8}
 & \multirow[t]{4}{*}{CIFAR100} & \multirow[t]{4}{*}{100} & 10 & \citet{tobaben2023Efficacy} & 0.933 & 0.788 & 0.525 \\
\cline{4-8}
 &  \citep{krizhevsky2009learning} &  & 25 & \citet{tobaben2023Efficacy} & 0.766 & 0.576 & 0.449 \\
\cline{4-8}
 &  &  & 50 & \citet{tobaben2023Efficacy} & 0.586 & 0.388 & 0.227 \\
\cline{4-8}
 &  &  & 100 & \citet{tobaben2023Efficacy} & 0.448 & 0.202 & 0.077 \\
\cline{2-8} \cline{3-8} \cline{4-8}
 & EuroSAT & 10 & 8 & This work & 0.791 & 0.388 & 0.144 \\
 & \citep{helber2019eurosat} & & &  &  & &  \\ 
\cline{2-8} \cline{3-8} \cline{4-8}
 & Patch Camelyon & 2 & 256 & This work & 0.379 & 0.164 & 0.076 \\
 & \citep{veeling2018rotation} & & &  &  & &  \\
\cline{2-8} \cline{3-8} \cline{4-8}
 & Pets & 37 & 8 & This work & 0.956 & 0.665 & 0.378 \\
 & \citep{parkhi2012cats} & & &  &  & &  \\
\cline{2-8} \cline{3-8} \cline{4-8}
 & Resics45 & 45 & 32 & This work & 0.632 & 0.379 & 0.217 \\
  & \citep{cheng2017remote} & & &  &  & &  \\
\cline{1-8} \cline{2-8} \cline{3-8} \cline{4-8}
\multirow[t]{2}{*}{from scratch } & CIFAR10 & 10 & 2500 & \citet{Carlini2022LiRA} & 0.300 & 0.110 & 0.084 \\
 & \citep{krizhevsky2009learning} & & &  &  & &  \\
\cline{2-8} \cline{3-8} \cline{4-8}
 (wide ResNet) & CIFAR100 & 100 & 250 & \citet{Carlini2022LiRA} & 0.700 & 0.400 & 0.276 \\
  & \citep{krizhevsky2009learning} & & &  &  & &  \\
\cline{1-8} \cline{2-8} \cline{3-8} \cline{4-8}
\bottomrule
\end{tabular}
\end{adjustbox}
\end{table}

\clearpage
\subsubsection{Predicting dataset vulnerability as function of $\shots$ and $\classes$}\label{sec:additional_glm_data}

This section provides additional results for the model based on \cref{eq:mia_vul_dataset}

\begin{table}[ht]
  \centering
  \caption{Results for fitting \cref{eq:mia_vul_dataset} with statsmodels~\cite{seabold2010statsmodels} to ViT Head data at $\fpr$ $\in \{0.1, 0.01, 0.001\}$. We utilize an ordinary least squares. The test $R^2$ assesses the fit to the data of R-50 Head.}
\begin{adjustbox}{max width=\textwidth}
\begin{tabular}{lrrrrrrrrr}
\toprule
\textbf{\boldmath coeff.} & \textbf{\boldmath $\fpr$} & \textbf{\boldmath $R^2$} & \textbf{\boldmath test $R^2$} & \textbf{coeff. value} & \textbf{std. error} & \textbf{\boldmath $t$} & \textbf{\boldmath $p>|z|$} & \textbf{\boldmath coeff. $[0.025$} & \textbf{\boldmath coeff. $0.975]$} \\
\midrule
\multirow[t]{5}{*}{$\beta_S$ (for $S$)} & 0.1 & 0.952 & 0.907 & -0.506 & 0.011 & -44.936 & 0.000 & -0.529 & -0.484 \\
 & 0.01 & 0.946 & 0.854 & -0.555 & 0.014 & -39.788 & 0.000 & -0.582 & -0.527 \\
 & 0.001 & 0.930 & 0.790 & -0.627 & 0.019 & -32.722 & 0.000 & -0.664 & -0.589 \\
\cline{1-10}
\multirow[t]{5}{*}{$\beta_C$ (for $C$)} & 0.1 & 0.952 & 0.907 & 0.090 & 0.021 & 4.231 & 0.000 & 0.048 & 0.131 \\
 & 0.01 & 0.946 & 0.854 & 0.182 & 0.026 & 6.960 & 0.000 & 0.131 & 0.234 \\
 & 0.001 & 0.930 & 0.790 & 0.300 & 0.036 & 8.335 & 0.000 & 0.229 & 0.371 \\
\cline{1-10}
\multirow[t]{5}{*}{$\beta_0$ (intercept)} & 0.1 & 0.952 & 0.907 & 0.314 & 0.045 & 6.953 & 0.000 & 0.225 & 0.402 \\
 & 0.01 & 0.946 & 0.854 & 0.083 & 0.056 & 1.491 & 0.137 & -0.027 & 0.193 \\
 & 0.001 & 0.930 & 0.790 & -0.173 & 0.077 & -2.261 & 0.025 & -0.324 & -0.022 \\
\cline{1-10}
\bottomrule
\end{tabular}
\label{tab:glm-coefs-diff}%
\end{adjustbox}
\end{table}%

\cref{fig:predict_mia_dataset_eval_additional} shows the performance for all considered $\fpr$.

\begin{figure*}[h]
    \centering
    \begin{subfigure}[b]{0.49\textwidth}
        \includegraphics[width=\textwidth]{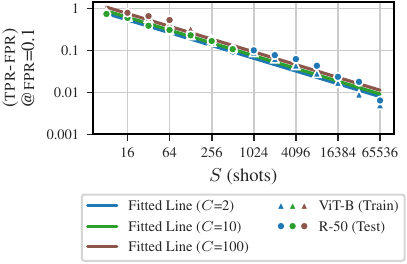}
        \caption{($(\tpr-\fpr)$ at $\fpr=0.1$}
    \end{subfigure}
    \begin{subfigure}[b]{0.49\textwidth}
        \includegraphics[width=\textwidth]{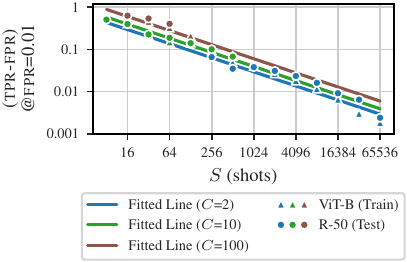}
        \caption{$(\tpr-\fpr)$ at $\fpr=0.01$}
    \end{subfigure}

    \vspace{0.5cm}
    
    \begin{subfigure}[b]{0.49\textwidth}
        \includegraphics[width=\textwidth]{figures/section4/glm_fit_show_prediction_fpr_0.001_difference.pdf}
        \caption{$(\tpr-\fpr)$ at $\fpr=0.001$}
    \end{subfigure}

    \caption{Predicted MIA vulnerability as a function of $\shots$ (shots) using a model based on \cref{eq:mia_vul_dataset} fitted \cref{tab:vit-function-shots} (ViT-B). The triangles show the median $\tpr-\fpr$ for the train set (ViT-B; \cref{tab:vit-function-shots}) and circle the test set (R-50; \cref{tab:r-50-function-shots}) over six seeds. Note that the triangles and dots for $\classes=10$ are for EuroSAT.}
    \label{fig:predict_mia_dataset_eval_additional}
\end{figure*}

\clearpage

\subsection{Simpler variant of the prediction model}\label{sec:predicting_dataset_simpler}

The prediction model in the main text (\cref{eq:mia_vul_dataset}) avoids predicting $\tpr < \fpr$ in the tail when $\shots$ is very large. In this section, we analyze a variation of the regression model that is simpler and predicts $\log_{10}(\tpr)$ instead of $\log_{10}(\tpr-\fpr)$. This variation fits worse to the empirical data and will predict $\tpr < \fpr$ for high $\shots$.

The general form this variant can be found in \cref{eq:mia_vul_dataset_variant}, where $\beta_{\shots}, \beta_{\classes}$ and $\beta_{0}$ are the learnable regression parameters.

\begin{equation}\label{eq:mia_vul_dataset_variant}
\log_{10}(\tpr) = \beta_{\shots}\log_{10}(\shots)+\beta_{\classes}\log_{10}(\classes)+\beta_{0}
\end{equation}

\cref{tab:glm-coefs-simpler} provides tabular results on the performance of the variant.

\begin{table}[ht]
  \centering
  \caption{Results for fitting \cref{eq:mia_vul_dataset_variant} with statsmodels~\cite{seabold2010statsmodels} to ViT Head data at $\fpr$ $\in \{0.1, 0.01, 0.001\}$. We utilize an ordinary least squares. The test $R^2$ assesses the fit to the data of R-50 Head.}
\begin{adjustbox}{max width=\textwidth}
\begin{tabular}{lrrrrrrrrr}
\toprule
\textbf{\boldmath coeff.} & \textbf{\boldmath $\fpr$} & \textbf{\boldmath $R^2$} & \textbf{\boldmath test $R^2$} & \textbf{coeff. value} & \textbf{std. error} & \textbf{\boldmath $t$} & \textbf{\boldmath $p>|z|$} & \textbf{\boldmath coeff. $[0.025$} & \textbf{\boldmath coeff. $0.975]$} \\
\midrule
\multirow[t]{5}{*}{$\beta_S$ (for $S$)} & 0.1 & 0.908 & 0.764 & -0.248 & 0.008 & -30.976 & 0.000 & -0.264 & -0.233 \\
 & 0.01 & 0.940 & 0.761 & -0.416 & 0.011 & -36.706 & 0.000 & -0.438 & -0.393 \\
 & 0.001 & 0.931 & 0.782 & -0.553 & 0.017 & -32.507 & 0.000 & -0.586 & -0.519 \\
\cline{1-10}
\multirow[t]{5}{*}{$\beta_C$ (for $C$)} & 0.1 & 0.908 & 0.764 & 0.060 & 0.015 & 3.955 & 0.000 & 0.030 & 0.089 \\
 & 0.01 & 0.940 & 0.761 & 0.169 & 0.021 & 7.941 & 0.000 & 0.127 & 0.211 \\
 & 0.001 & 0.931 & 0.782 & 0.297 & 0.032 & 9.303 & 0.000 & 0.234 & 0.360 \\
\cline{1-10}
\multirow[t]{5}{*}{$\beta_0$ (intercept)} & 0.1 & 0.908 & 0.764 & 0.029 & 0.032 & 0.913 & 0.362 & -0.034 & 0.093 \\
 & 0.01 & 0.940 & 0.761 & -0.118 & 0.045 & -2.613 & 0.010 & -0.208 & -0.029 \\
 & 0.001 & 0.931 & 0.782 & -0.295 & 0.068 & -4.345 & 0.000 & -0.429 & -0.161 \\
\cline{1-10}
\bottomrule
\end{tabular}
\label{tab:glm-coefs-simpler}%
\end{adjustbox}
\end{table}%

\cref{fig:predict_mia_dataset_eval_simpler_additional} plots the performance of the variant similar to \cref{fig:predict_mia_dataset_eval} in the main text.

\begin{figure*}[h]
    \centering
    \begin{subfigure}[b]{0.49\textwidth}
        \includegraphics[width=\textwidth]{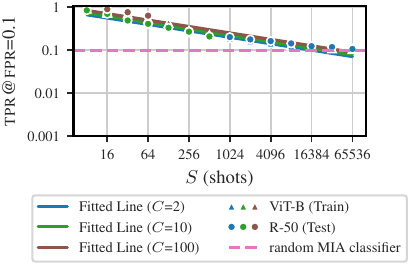}
        \caption{$\tpr$ at $\fpr=0.1$}
    \end{subfigure}
    \begin{subfigure}[b]{0.49\textwidth}
        \includegraphics[width=\textwidth]{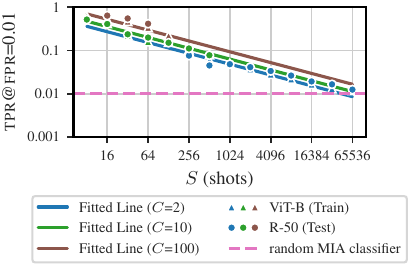}
        \caption{$\tpr$ at $\fpr=0.01$}
    \end{subfigure}

    \vspace{0.5cm}
    
    \begin{subfigure}[b]{0.49\textwidth}
        \includegraphics[width=\textwidth]{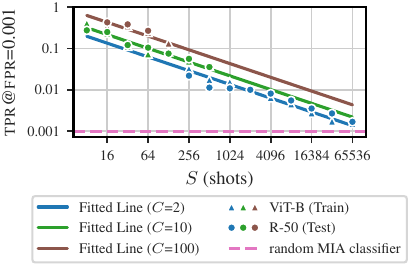}
        \caption{$\tpr$ at $\fpr=0.001$}
    \end{subfigure}
    \caption{Predicted MIA vulnerability as a function of $\shots$ (shots) using a model based on \cref{eq:mia_vul_dataset_variant} fitted \cref{tab:vit-function-shots} (ViT-B). The triangles show the median $\tpr$ for the train set (ViT-B; \cref{tab:vit-function-shots}) and circle the test set (R-50; \cref{tab:r-50-function-shots}) over six seeds. Note that the triangles and dots for $\classes=10$ are for EuroSAT.}
    \label{fig:predict_mia_dataset_eval_simpler_additional}
\end{figure*}

\clearpage
\subsection{Empirical results for RMIA}
\cref{fig:rmia_function_shots,fig:rmia_comparison_to_lira,fig:rmia-lira-comparison_additional} report additional results for RMIA~\cite{Zarifzadeh2024RMIA}.

\begin{figure}[h!]
    \begin{subfigure}[b]{0.49\textwidth}
        \includegraphics[width=\textwidth]{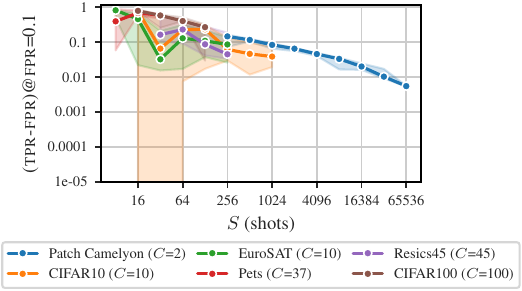}
        \caption{$(\tpr-\fpr)$ at $\fpr=0.1$}
    \end{subfigure}
    \begin{subfigure}[b]{0.49\textwidth}
        \includegraphics[width=\textwidth]{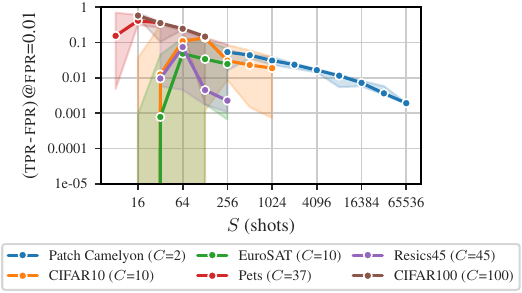}
        \caption{$(\tpr-\fpr)$ at $\fpr=0.01$}
    \end{subfigure}
    \begin{subfigure}[b]{0.49\textwidth}
        \includegraphics[width=\textwidth]{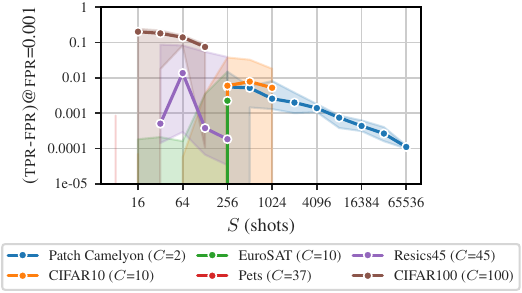}
        \caption{$(\tpr-\fpr)$ at $\fpr=0.001$}
    \end{subfigure}
    \caption{RMIA~\citep{Zarifzadeh2024RMIA} vulnerability ($\tpr-\fpr$ at fixed $\fpr$) as a function of $S$ (shots) when attacking a ViT-B Head fine-tuned without DP on different datasets. We observe at power-law relationship but especially at low $\fpr$ the relationship is not as clear as with LiRA (compare to \cref{fig:mia_function_shots_additional}). The solid line displays the median and the error bars the minimum of the lower bounds and maximum of the upper bounds for the Clopper-Pearson CIs over six seeds.}
    \label{fig:rmia_function_shots}
\end{figure}

\begin{figure}[h!]
    \begin{subfigure}[b]{0.49\textwidth}
        \includegraphics[width=\textwidth]{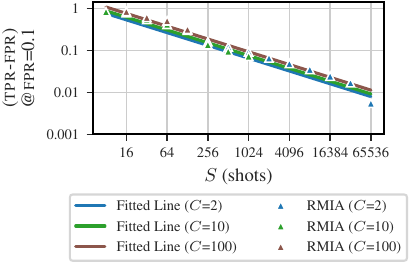}
        \caption{$(\tpr-\fpr)$ at $\fpr=0.1$}
    \end{subfigure}
    \begin{subfigure}[b]{0.49\textwidth}
        \includegraphics[width=\textwidth]{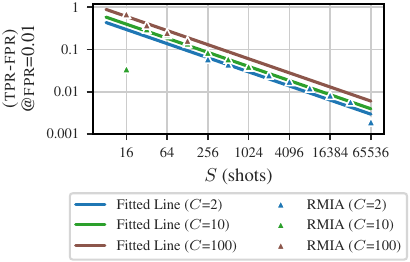}
        \caption{$(\tpr-\fpr)$ at $\fpr=0.01$}
    \end{subfigure}
    \begin{subfigure}[b]{0.49\textwidth}
        \includegraphics[width=\textwidth]{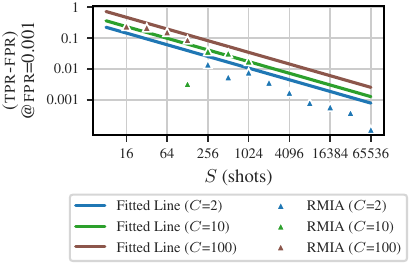}
        \caption{$(\tpr-\fpr)$ at $\fpr=0.001$}
    \end{subfigure}
    \caption{Predicted MIA vulnerability ($(\tpr-\fpr)$ at $\fpr$) based on LiRA vulnerability data as a function of $\shots$ (shots) in comparison to observed RMIA~\citep{Zarifzadeh2024RMIA} vulnerability on the same settings. The triangles show the highest $\tpr$ when attacking (ViT-B Head) with RMIA over six seeds (datasets: Patch Camelyon, EuroSAT and CIFAR100). Especially at $\fpr=0.1$ the relationship behaves very similar for both MIAs, but RMIA shows more noisy behavior at lower $\fpr$.}
    \label{fig:rmia_comparison_to_lira}
\end{figure}

\begin{figure}[h]
\begin{subfigure}[b]{\textwidth}
        \includegraphics[width=\textwidth]{figures/section4/RMIA_LiRA_as_function_of_shots_0.1_ViT-B_difference_median.pdf}
        \caption{$(\tpr-\fpr)$ at $\fpr=0.1$}
    \end{subfigure}

    \vspace{0.5cm}
    
    \begin{subfigure}[b]{\textwidth}
        \includegraphics[width=\textwidth]{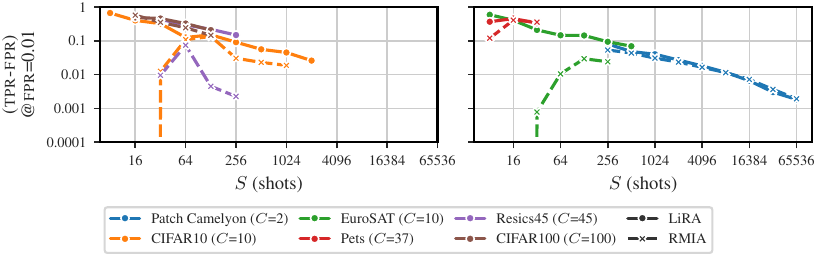}
        \caption{$(\tpr-\fpr)$ at $\fpr=0.01$}
    \end{subfigure}

    \vspace{0.5cm}
    
    \begin{subfigure}[b]{\textwidth}
        \includegraphics[width=\textwidth]{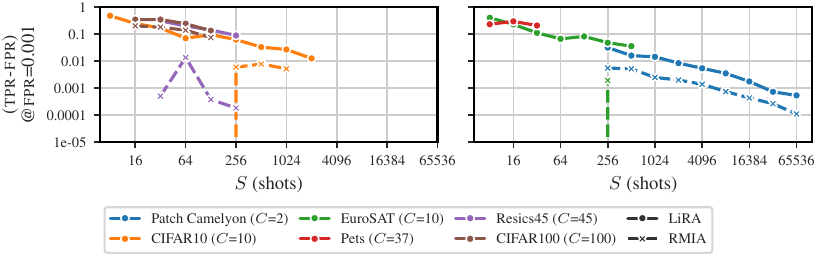}
        \caption{$(\tpr-\fpr)$ at $\fpr=0.001$}
    \end{subfigure}
    \caption{LiRA and RMIA vulnerability ($(\tpr-\fpr)$) as a function of shots ($\shots$) when attacking a ViT-B Head fine-tuned without DP on different datasets. For better visibility, we split the datasets into two panels. We observe the power-law for both attacks, but the RMIA is more unstable than LiRA. The lines display the median over six seeds.}
    \label{fig:rmia-lira-comparison_additional}
\end{figure}

\clearpage
\subsection{Tabular results for Section \ref{sec:individual_mia}}

\cref{tab:individual_mia} displays the tabular results for \cref{fig:individual_mia} in \cref{sec:individual_mia}.
 
\begin{table}[h]
    \caption{Tabular results for \cref{fig:individual_mia} on when attacking a ViT-B (Head) fine-tuned on PatchCamelyon. We display the median over six seeds at $\fpr = 0.1$.}
    \label{tab:individual_mia}
    \centering
    \begin{tabular}{rrrrr}
    \toprule
    $S$ & Max $\tpr$ & $\tpr$ of 0.999 Quantile & $\tpr$ of 0.99 Quantile & $\tpr$ of 0.95 Quantile \\
    \midrule
    16384 & 0.77 & 0.43 & 0.22 & 0.11 \\
    23170 & 0.73 & 0.37 & 0.19 & 0.10 \\
    32768 & 0.65 & 0.30 & 0.15 & 0.08 \\
    49152 & 0.54 & 0.23 & 0.13 & 0.07 \\
    65536 & 0.49 & 0.19 & 0.11 & 0.07 \\
    \bottomrule
    \end{tabular}
\end{table}

\clearpage

\section{Details on Section \ref{sec:bound_comparision}}\label{app:bound_comparison}

In \cref{sec:bound_comparision}, we compare the results of our empirical models of vulnerability (\cref{sec:model-to-predict-vulnerability,sec:individual_mia}) to DP bounds. Below we explain how we make the connection between both.

First, we compute the upper bound on the $\tpr$ for a given $(\epsilon, \delta)$-DP privacy budget at a given $\fpr$ using \cref{thm:Kairouz_bound} reformulated from  \citet{Kairouz2015Composition} below:

\begin{theorem}[\citet{Kairouz2015Composition}]
A mechanism $\mathcal{M}: \mathcal{X} \to \mathcal{Y}$ is $(\epsilon, \delta)$-DP if and only if for all adjacent $\mathcal{D}\sim\mathcal{D}'$
\begin{equation}
\label{eq:tpr_bound}
    \tpr \leq \min\{ e^\epsilon \fpr + \delta,
    1 - e^{-\epsilon} (1 - \delta - \fpr) \} \ .
\end{equation}
\label{thm:Kairouz_bound}
\end{theorem}

For a given $(\epsilon,\delta)$ and $\fpr$ we then obtain a value for the $\tpr$.

Next, we use the linear model from \cref{sec:model-to-predict-vulnerability} to solve for the minimum $S$ predicted to be required given $C=2$ classes in our example. The coefficients can be found in \cref{tab:glm-coefs-diff} for the average case and \cref{sec:individual_mia} for the worst-case. We solve the $\tpr$ from the linear model as
\begin{align}
    &\log_{10}(\tpr-\fpr) = \beta_{\shots}\log_{10}(\shots)+\beta_{\classes}\log_{10}(\classes)+\beta_{0}\\
    \Leftrightarrow
    &\tpr = \shots ^ {\beta_{\shots}} \classes ^ {\beta_{\classes}} 10^{\beta_{0}} + \fpr \label{eq:linear-model-tpr-avg}.
\end{align}

Now, we find the minimum $\shots$ that the $\tpr$ from \cref{eq:linear-model-tpr-avg} upper bounds the $\tpr$ of \cref{eq:tpr_bound} as
\begin{align}
    &\shots ^ {\beta_{\shots}} \classes ^ {\beta_{\classes}} 10^{\beta_{0}} + \fpr 
    = \min\{ e^\epsilon \fpr + \delta, 1 - e^{-\epsilon} (1 - \delta - \fpr) \} \\
    \Rightarrow
    & S
    = \left( \frac{\min\{ e^\epsilon \fpr + \delta, 1 - e^{-\epsilon} (1 - \delta - \fpr) \} - \fpr}{\classes^{\beta_{\classes}}10^{\beta_0}} \right) ^{1/{\beta_\shots}}
\end{align}

\clearpage

\section*{NeurIPS Paper Checklist}

\begin{enumerate}

\item {\bf Claims}
    \item[] Question: Do the main claims made in the abstract and introduction accurately reflect the paper's contributions and scope?
    \item[] Answer: \answerYes{} 
    \item[] Justification: The main claims that are summarized at the end of \cref{sec:introduction} accuracly reflect the contributions in \cref{sec:explaining_trends,sec:predicting_dataset}.

\item {\bf Limitations}
    \item[] Question: Does the paper discuss the limitations of the work performed by the authors?
    \item[] Answer: \answerYes{} 
    \item[] Justification: The limiations are discussed in \cref{sec:discussion}.

\item {\bf Theory assumptions and proofs}
    \item[] Question: For each theoretical result, does the paper provide the full set of assumptions and a complete (and correct) proof?
    \item[] Answer: \answerYes{} 
    \item[] Justification: The details for the proofs in \cref{sec:explaining_trends} are in \cref{sec:theory_details}.

    \item {\bf Experimental result reproducibility}
    \item[] Question: Does the paper fully disclose all the information needed to reproduce the main experimental results of the paper to the extent that it affects the main claims and/or conclusions of the paper (regardless of whether the code and data are provided or not)?
    \item[] Answer: \answerYes{} 
    \item[] Justification: Experimental details are in \cref{sec:training_details} and the documented source code is in the supplementary material.

\item {\bf Open access to data and code}
    \item[] Question: Does the paper provide open access to the data and code, with sufficient instructions to faithfully reproduce the main experimental results, as described in supplemental material?
    \item[] Answer: \answerYes{} 
    \item[] Justification: The source code including instructions can be found in the supplementary material.

\item {\bf Experimental setting/details}
    \item[] Question: Does the paper specify all the training and test details (e.g., data splits, hyperparameters, how they were chosen, type of optimizer, etc.) necessary to understand the results?
    \item[] Answer: \answerYes{}{} 
    \item[] Justification: We provide details in \cref{sec:training_details} and the remaining details can be found in the documented source code.

\item {\bf Experiment statistical significance}
    \item[] Question: Does the paper report error bars suitably and correctly defined or other appropriate information about the statistical significance of the experiments?
    \item[] Answer: \answerYes{}{} 
    \item[] Justification: In \cref{sec:predicting_dataset} we provide results over at least six seeds and plot the median as well as errorbars. Many errorbars are based on Clopper-Pearson intervals.

\item {\bf Experiments compute resources}
    \item[] Question: For each experiment, does the paper provide sufficient information on the computer resources (type of compute workers, memory, time of execution) needed to reproduce the experiments?
    \item[] Answer: \answerYes{}
    \item[] Justification: Details can be found in \cref{sec:compute_resources}.
    
\item {\bf Code of ethics}
    \item[] Question: Does the research conducted in the paper conform, in every respect, with the NeurIPS Code of Ethics \url{https://neurips.cc/public/EthicsGuidelines}?
    \item[] Answer: \answerYes{} 
    \item[] Justification: The research conforms with the NeurIPS Code of Ethics.

\item {\bf Broader impacts}
    \item[] Question: Does the paper discuss both potential positive societal impacts and negative societal impacts of the work performed?
    \item[] Answer: \answerYes{} 
    \item[] Justification: We discuss the broader impacts in \cref{sec:discussion}.
    
\item {\bf Safeguards}
    \item[] Question: Does the paper describe safeguards that have been put in place for responsible release of data or models that have a high risk for misuse (e.g., pretrained language models, image generators, or scraped datasets)?
    \item[] Answer: \answerNA{} 
    \item[] Justification: The paper does not pose such risks.

\item {\bf Licenses for existing assets}
    \item[] Question: Are the creators or original owners of assets (e.g., code, data, models), used in the paper, properly credited and are the license and terms of use explicitly mentioned and properly respected?
    \item[] Answer: \answerYes{}{} 
    \item[] Justification: We cite the original owners and mention the licenses in \cref{sec:dataset_licenses,sec:checkpoint_licenses}.

\item {\bf New assets}
    \item[] Question: Are new assets introduced in the paper well documented and is the documentation provided alongside the assets?
    \item[] Answer: \answerYes{} 
    \item[] Justification: The documentation for the code is with the code in the supplementary material.

\item {\bf Crowdsourcing and research with human subjects}
    \item[] Question: For crowdsourcing experiments and research with human subjects, does the paper include the full text of instructions given to participants and screenshots, if applicable, as well as details about compensation (if any)? 
    \item[] Answer: \answerNA{} 
    \item[] Justification: The paper does not involve crowdsourcing nor research with human subjects.

\item {\bf Institutional review board (IRB) approvals or equivalent for research with human subjects}
    \item[] Question: Does the paper describe potential risks incurred by study participants, whether such risks were disclosed to the subjects, and whether Institutional Review Board (IRB) approvals (or an equivalent approval/review based on the requirements of your country or institution) were obtained?
    \item[] Answer: \answerNA{} 
    \item[] Justification: The paper does not involve crowdsourcing nor research with human subjects.

\item {\bf Declaration of LLM usage}
    \item[] Question: Does the paper describe the usage of LLMs if it is an important, original, or non-standard component of the core methods in this research? Note that if the LLM is used only for writing, editing, or formatting purposes and does not impact the core methodology, scientific rigorousness, or originality of the research, declaration is not required.
    \item[] Answer: \answerNA{}{} 
    \item[] Justification: The core method development in this research does not involve LLMs as any important, original, or non-standard components

\end{enumerate}

\end{document}